\documentclass[12pt,transaction, draftclsnofoot,onecolumn]{IEEEtran}
\usepackage{amsfonts}
\usepackage{amssymb}
\usepackage{hyperref}
\usepackage{graphicx}
\usepackage{subfigure}
\usepackage{enumerate}
\usepackage{amsmath}
\usepackage{color}
\usepackage{amsthm}
\usepackage{amsmath}
\usepackage{algorithm}
\usepackage{algpseudocode}
\usepackage{color}
\usepackage{accents}
\usepackage{lipsum}
\usepackage{float}
\usepackage{breakurl}
\usepackage[font={small}]{caption}
\algtext*{EndWhile}
\algtext*{EndIf}
\algtext*{EndFor}

\newcommand\ubar[1]{%
	\underaccent{\bar}{#1}}
\usepackage[hang,flushmargin]{footmisc}
\usepackage{caption}

\DeclareMathOperator*{\argmin}{arg\,min}

\newcommand{\bp}{\begin{proof} \small }
\newcommand{\ep}{\end{proof} \normalsize}
\newcommand{\epx}{\end{proof} \small}
\newcommand{\bpa}{\begin{proofappx} \footnotesize }
\newcommand{\epa}{\end{proofappx} \small }
\newtheorem{theorem}{Theorem}

\newtheorem{lemma}{Lemma}
\newtheorem{assumption}{Assumption}
\newtheorem{definition}{Definition}

\newtheorem*{theorem*}{Theorem}
\newtheorem*{proposition*}{Proposition}
\newtheorem*{corollary*}{Corollary}
\newtheorem*{lemma*}{Lemma}
\newtheorem*{assumption*}{Assumption}
\newtheorem*{definition*}{Definition}
\newtheorem*{claim*}{Claim}

\newcommand{\be}{\begin{equation}}
\newcommand{\ee}{\end{equation}}
\newcommand{\bs}{\begin{subequations}}
\newcommand{\es}{\end{subequations}}
\newcommand{\bq}{\begin{eqnarray}}
\newcommand{\eq}{\end{eqnarray}}
\newcommand{\bqn}{\begin{eqnarray*}}
\newcommand{\eqn}{\end{eqnarray*}}

\newcommand{\ba}{\left[ \begin{array}}
\newcommand{\ea}{\\ \end{array} \right]}
\newcommand{\ben}{\begin{enumerate}}
\newcommand{\een}{\end{enumerate}}

\def\P{{\boldsymbol{P}}}

\def\X{{\boldsymbol{X}}}

\def\d{{\boldsymbol{d}}}

\def\p{{\boldsymbol{p}}}

\def\x{{\boldsymbol{x}}}

\def\real{{\mathchoice%
{\hbox{\rm\setbox1=\hbox{I}\copy1\kern-.45\wd1 R}}
{\hbox{\rm\setbox1=\hbox{I}\copy1\kern-.45\wd1 R}}
{\hbox{\scriptsize\rm\setbox1=\hbox{I}\copy1\kern-.45\wd1 R}}
{\hbox{\scriptsize\rm\setbox1=\hbox{I}\copy1\kern-.45\wd1 R}}}}

\def\Zint{{\mathchoice{\setbox1=\hbox{\sf Z}\copy1\kern-.75\wd1\box1}
{\setbox1=\hbox{\sf Z}\copy1\kern-.75\wd1\box1}
{\setbox1=\hbox{\scriptsize\sf Z}\copy1\kern-.75\wd1\box1}
{\setbox1=\hbox{\scriptsize\sf Z}\copy1\kern-.75\wd1\box1}}}
\newcommand{\complex}{ \hbox{\rm C\kern-0.45em\rule[.07em]{.02em}{.58em}%
\kern 0.43em}}

\DeclareMathOperator*{\argmax}{arg\,max} 

\ifodd 1

\else

\fi

\begin{document}
%
\title{Spatio-temporal Edge Service Placement: A Bandit Learning Approach}
%
\author{Lixing~Chen,~\IEEEmembership{Student~Member,~IEEE},
        ~Jie~Xu,~\IEEEmembership{Member,~IEEE},
        ~Shaolei~Ren,~\IEEEmembership{Member,~IEEE},
        ~Pan~Zhou,~\IEEEmembership{Member,~IEEE}
\thanks{L. Chen and J. Xu are with Department of Electrical and
	Computer Engineering, University of Miami, USA. }
\thanks{S. Ren is with Department of Electrical and Computer Engineering, University of California, Riverside, USA. }
\thanks{P. Zhou is with School of EIC, Huazhong University of Science and Technology, China. }
}
\maketitle
\vspace{-0.3 in}
\begin{abstract}
Shared edge computing platforms deployed at the radio access network are expected to significantly improve quality of service delivered by Application Service Providers (ASPs) in a flexible and economic way. However, placing edge service in every possible edge site by an ASP is practically infeasible due to the ASP's prohibitive budget requirement. In this paper, we investigate the edge service placement problem of an ASP under a limited budget, where the ASP dynamically rents computing/storage resources in edge sites to host its applications in close proximity to end users. Since the benefit of placing edge service in a specific site is usually unknown to the ASP a priori, optimal placement decisions must be made while learning this benefit. We pose this problem as a novel combinatorial contextual bandit learning problem. It is ``combinatorial'' because only a limited number of edge sites can be rented to provide the edge service given the ASP's budget. It is ``contextual'' because we utilize user context information to enable finer-grained learning and decision making. To solve this problem and optimize the edge computing performance, we propose SEEN, a Spatial-temporal Edge sErvice placemeNt algorithm. Furthermore, SEEN is extended to scenarios with overlapping service coverage by incorporating a disjunctively constrained knapsack problem. In both cases, we prove that our algorithm achieves a sublinear regret bound when it is compared to an oracle algorithm that knows the exact benefit information. Simulations are carried out on a real-world dataset, whose results show that SEEN significantly outperforms benchmark solutions.
\end{abstract}


%
\IEEEpeerreviewmaketitle

\section{Introduction}
Mobile cloud computing (MCC) supports mobile applications in resource-constrained mobile devices by offloading computation-demanding tasks to the resource-rich remote cloud. Intelligent personal assistant applications are perhaps the most popular applications that rely on MCC, where the speech recognition engine that uses advanced machine learning technologies to function resides in the cloud server. Nowadays, mobile applications such as virtual/augmented reality and mobile gaming are becoming even more data-hungry, latency-sensitive and location-aware. For example, Google Lens is a real-time image recognition application that can pull up the right information (e.g. restaurant reviews and menus) in an interactive and responsive way as the user points his/her smartphone camera to objects (e.g. a restaurant) when passing by. However, as these applications become more prevalent, ensuring high quality of service (QoS) becomes very challenging due to the backbone network congestion, delay and expensive bandwidth \cite{li2010cloudcmp}.

To address these challenges, mobile edge computing (MEC) \cite{taleb20175G} has recently been proposed. The key idea of MEC is to move the computation resources to the logical edge of the Internet, thereby enabling analytics and knowledge generation to occur closer to the data sources. Such an edge service provisioning scenario is no longer a mere vision, but becoming a reality. Vapor IO \cite{VaporIO} has launched Project Volutus \cite{Volutus} to deliver shared edge computing services via a network of micro data centers deployed in cellular tower sites. In a recent white paper \cite{Intel}, Intel envisions that its smart cell platform will allow mobile operators to sell IT real-estate at the radio access network and to monetize their most valuable assets without compromising any of the network features. It is anticipated that Application Service Providers (ASPs) will soon be able to rent computation resources in such \textit{shared} edge computing platforms in a flexible and economic way. Fig. \ref{fig:googlelens} illustrates how the Google Lens application can leverage the shared edge computing platform to improve QoS.
\begin{figure} [b]
	\vspace{-0.3 in}
	\centering
	\includegraphics[width=0.7\linewidth]{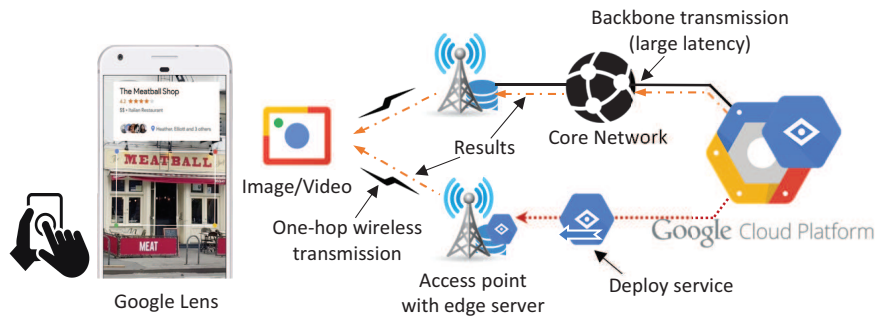}
	\caption{Mobile Edge Computing for Google Lens. In MCC (upper side), the image/video is sent to cloud server for processing via backbone Internet and incur large transmission latency; In MEC (lower side), the service is deployed at edge server collocated at the access point, and hence the image/video only needs to be sent to edge server for processing via one-hope wireless connection, which enables quick response.}
	\label{fig:googlelens}
	\vspace{-0.2 in}
\end{figure}

While provisioning edge service in every possible edge site (i.e. base station) can deliver the best QoS, it is practically infeasible, especially for small and starting ASPs, due to the prohibitive budget requirement. In common business practice, an ASP has a budget on the operation expenses in mind and desires the best performance within the budget \cite{stansberry2013uptime}. This means that an ASP will only be able to place edge services in a limited number of edge sites and hence, where to rent edge computation resources must be judiciously decided to improve QoS under the limited budget. 

Deciding the optimal edge service placement policy faces significant challenges due to information uncertainty in both spatial and temporal domains. Firstly, the benefit of edge service provisioning primarily depends on the service demand of users, which can vary considerably across different sites. However, the demand pattern is usually unknown to the ASP before deploying edge service in a particular site, and may also vary substantially after frequent application updates. Because the demand pattern can only be observed at sites where the edge service is deployed, how to make the optimal tradeoff between exploration (i.e. to place edge service in unknown sites to learn the user demand pattern) and exploitation (i.e. to place edge service at high-demanding sites to maximize ASP utility) is a key challenge. Secondly, even in the same site, the service demand varies over time depending on who are currently in the mobile cell, what their preferences are, what mobile devices they use, time and other environmental variables. Collectively, this information is called the context information. Incorporating this valuable information into the edge service placement decision making, in addition to the plain number of devices, is likely to improve the overall system performance but is challenging because the context space can be huge. Learning for each specific context is nearly impossible due to the limited number of context occurrences. A promising approach is to group similar contexts so that learning can be carried out on the context-group level. However, how to group contexts in a way that enables both fast and effective learning demands for a careful design.

In this paper, we study the spatial-temporal edge service placement problem of an ASP under a limited budget and propose an efficient learning algorithm, called SEEN (Spatial-temporal Edge sErvice placemeNt), to optimize the edge computing performance. SEEN does not assume a priori knowledge about users' service demand. Rather, it learns the demand pattern in an online fashion by observing the realized demand in sites where edge service is provisioned and uses this information to make future edge service placement decisions. In particular, SEEN is location-aware as it uses information only in the local area for each base station and is context-aware as it utilizes user context information to make edge service placement decisions.

The spatial-temporal edge service placement problem is posed as a novel Contextual Combinatorial Multi-armed Bandit (CC-MAB) problem \cite{lai1985asymptotically} (see more detailed literature review in Section II). We analytically bound the loss due to learning, termed regret, of SEEN compared to an oracle benchmark that knows precisely the user demand pattern a priori. A sublinear regret bound is proved, which not only implies that SEEN produces an asymptotically optimal edge service placement policy but also provides finite-time performance guarantee. The proposed algorithm is further extended to scenarios with overlapping service coverage. In this case, a disjunctively constrained knapsack problem is incorporated into the framework of SEEN to deal with the service demand coupling caused by the coverage overlapping among cells. We prove that the sublinear regret bound still holds. To evaluate the performance of SEEN, we carry out extensive simulations on a real-world dataset on mobile application user demand \cite{lim2015investigating}, whose results show that SEEN significantly outperforms benchmark algorithms.

The rest of this paper is organized as follows. Section \ref{sec:related_work} reviews related works. Section \ref{sec:system model} presents the system model and formulates the problem. Section \ref{sec:SEEN} designs SEEN and analyzes its performance. Section \ref{sec:overlapping} extends SEEN to the overlapping coverage scenario. Section \ref{sec:simulation} presents the simulation results, followed by the conclusion in Section \ref{sec:conclusion}.

\section{Related Work}\label{sec:related_work}
Mobile edge computing has attracted much attention in recent years \cite{mao2017mobile,shi2016edge}. Many prior studies focus on computation offloading, concerning what/when/how to offload users' workload from their devices to the edge servers or the cloud. Various works have studied different aspects of this problem, considering e.g. stochastic task arrivals \cite{huang2012dynamic,liu2016delay}, energy efficiency \cite{xu2017online,mao2016dynamic}, collaborative offloading \cite{chen2017socially,tanzil2016distributed}, etc. However, these works focus on the optimization problem after certain edge services have been provisioned at the Internet edge. By contrast, this paper focuses on how to place edge service among many possible sites in an edge system.

Service placement in edge computing has been studied in many contexts in the past. Considering content delivery as a service, many prior works study placing content replicas in traditional content delivery networks (CDNs) \cite{chen2002dynamic} and, more recently, in wireless caching systems such as small cell networks \cite{shanmugam2013femtocaching}. Early works addressed the centralized cases where the demand profile is static or time invariant, and dynamic service placement in geographically distributed clouds is studied in \cite{zhang2013dynamic} in the presence of demand and resource dynamics. Our prior works \cite{xu2018joint,chen2017collaborative} investigate collaborative service placement to improve the efficiency of edge resource utilization by enabling cooperation among edge servers. However, these works assume that the service demand is known a priori whereas the service demand pattern in our problem has to be learned over time. A learning-based content caching algorithm for a wireless caching node was recently developed in \cite{muller2017context}, which also takes a contextual bandit learning approach similar to ours. However, it considers the caching policies (i.e. which content to cache) in a single caching site whereas we aim to determine where to place edge service among multiple edge sites, which may have to maintain distinct context spaces. Importantly, we also consider the coupled decisions among multiple sites due to the possible overlapping coverage while content files in \cite{muller2017context} are treated independently.

MAB algorithms have been widely studied to address the critical tradeoff between exploration and exploitation in sequential decision making under uncertainty \cite{lai1985asymptotically}. The basic MAB setting concerns with learning the single optimal action among a set of candidate actions of a priori unknown rewards by sequentially trying one action each time and observing its realized noisy reward \cite{auer2002finite,agrawal1995sample}. Combinatorial bandits extends the basic MAB by allowing multiple-play each time (i.e. choosing multiple edge sites under a budget in our problem) \cite{anantharam1987asymptotically,agrawal1990multi,gai2012combinatorial} and contextual bandits extends the basic MAB by considering the context-dependent reward functions \cite{slivkins2011contextual,li2010contextual,tekin2015distributed}. While both combinatorial bandits and contextual bandits problems are already much more difficult than the basic MAB problem, this paper tackles the even more difficult CC-MAB problem. Recently, a few other works \cite{qin2014contextual,li2016contextual} also started to study CC-MAB problems. However, these works make strong assumptions that are not suitable for our problem. For instance, \cite{qin2014contextual,li2016contextual} assume that the reward of an individual action is a linear function of the contexts. \cite{muller2017context} is probably the most related work that investigates contextual and combinatorial MAB for proactive caching. However, our work has many key differences from \cite{muller2017context}. First, \cite{muller2017context} considers CC-MAB for a single learner (a caching station) and maintains a common context space for all users. By contrast, our paper considers a multi-learner case, where each learner (i.e. SBS) learns the demand pattern of users within its service range. More importantly, we allow each SBS to maintain a distinct location-specific context space and collect different context information of connected users according to the user’s preference. Second, while \cite{muller2017context} considers a bandit learning problem for a fixed size of content items, we allow our algorithm to deal with infinitely large user set. Third, we further consider an overlapped edge network and address the decision coupling among edge sites due to overlapped coverage.

\section{System Model}\label{sec:system model}
\subsection{Edge System and Edge Service Provisioning}
We consider a heterogeneous network consisting of $N$ small cells (SCs), indexed by $\mathcal{N}$, and a macro base station (MBS). Each SC has a small-cell base station (SBS) equipped with a shared edge computing platform and thus is an edge site that can be used to host edge services for ASPs. The MBS provides ubiquitous radio coverage and access to the cloud server in case edge computing is not accessible. SBSs (edge sites) provide Software-as-a-Service (SaaS) to ASPs, managing computation/storage resources (e.g. CPU, scheduling, etc.) to ensure end-to-end QoS, while the ASP maintains its own user data, serving as a middleman between end users and SaaS Providers. As such, SBSs charge the ASP for the amount of time the edge service is rented. Fig. \ref{fig:system_model} gives an illustration for the considered scenario. 
\begin{figure}[htb]
	\centering
	\includegraphics[width=0.7\linewidth]{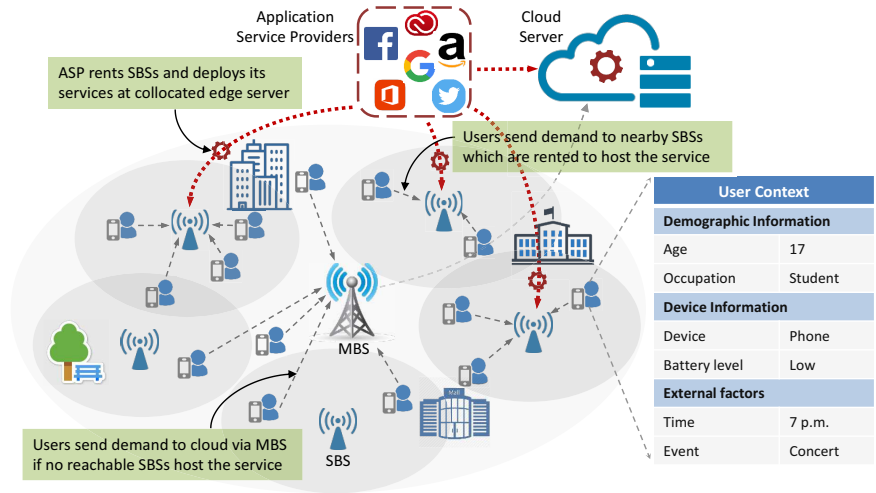}
	\caption{Illustration of context-aware edge service provisioning}
	\label{fig:system_model}
	\vspace{-0.2 in}
\end{figure}

Specifically, computation and storage resource allocation in SBSs can be realized by \emph{containerization} techniques \cite{pahl2015containerization}, e.g., Dockers and Kubernetes \cite{bernstein2014containers}. The key advantage of containerization over the virtual machine technology is that it incurs much lower system overhead and much shorter launch time. For example, each SBS can set up a \emph{Docker Registry} to store \emph{Dock images} (i.e. a package that encapsulates the running environment of an application) locally. When the SBS is chosen to host the ASP's application, it will pull up the Docker image for the corresponding application and configure the container in seconds \cite{russell2014kvm}.  Without loss of generality, this paper focuses on the service placement for one application. Due to the limited budget, the ASP can only rent up to $b (b < N)$ SBSs, where we assume for simplicity that all SBSs charge the same price for a unit time.
\begin{table}[htb]
		\centering
		\caption{Nomenclature}\label{tab:runtime}
		\begin{tabular}{l|l} 
			\hline
			\textbf{Variable} & \textbf{Description}  \\
			\hline
			$\mathcal{N}$      & a set of total $N$ SBSs \\
			$b$  & the budget of ASP \\
			$\mathcal{S}^t$ & the set of SBSs selected in slot $t$ \\
			$\mathcal{S}^{*t}$ &  Oracle solution in slot $t$ \\
			$\mathcal{M}^t$    & user population in slot $t$  \\
			$\mathcal{M}^t_n$    & users covered by SBS $n$   \\
			$d^t_m$    &  service demand of user $m$ in slot $t$   \\
			$\lambda$    &  input data size of one task  \\
			$\eta$    &  required CPU cycles for one task  \\
			$Q^t_{n,m}$    &  the delay of completing one task for user $m$ at SBS $n$  \\
			$\tilde{u}^t_{n,m}$    &  the delay reduction of one task\\
			$\mathcal{X}_n$    &  context space maintained by SBS $n$\\
			$D_n$    &  dimension of context space monitored by SBS $n$\\
			$\mathcal{P}_{n,T}$    & partition created on context space $\mathcal{X}_n$\\
			$x_{n,m}$    &  user $m$'s context observed by SBS $n$, $x_{n,m} \in \mathcal{X}_n$\\
			$\X^t$    &  contexts of all users in slot $t$, $\X^t= (x^t_{n,m})_{m\in\mathcal{M}_n, n\in\mathcal{N}}$ \\
			$\mu(x)$    &  expected service demand for a user with context $x$\\
			$\hat{d}(p)$    &  demand estimation for users with context in hypercube $p$\\
			\hline
		\end{tabular}\\
		\vspace{-0.1 in}
\end{table}

The operational timeline is discretized into time slots. In each time slot $t$, ASP chooses a set of SBSs $\mathcal{S}^t\in\mathcal{N}$, where $|\mathcal{S}^t|\leq b$, for application deployment. This decision is referred to as the (edge) service placement decision in the rest of this paper. Let $\mathcal{M}^t={1,\dots,M^t}$ be the user population served by the entire network in time slot $t$ and let $\mathcal{M}^t_n\subseteq\mathcal{M}^t$ be the user population covered by SBS $n$. The user population in the considered network can vary across time slots because of the user mobility. Users can also move within a time slot but we assume that the User-SBS association remains the same within a time slot for simplicity. We consider that the service placement decisions are made on the scale of minutes so that frequent reconfiguration of edge services is avoided while the temporal variation of user population is largely captured.

We will first consider the case where the service areas of SBSs are non-overlapping and then consider the case with overlapping service areas in Section \ref{sec:overlapping}. In the non-overlapping case, if SBS $n$ is chosen by the ASP to host the application in time slot $t$, i.e. $n\in\mathcal{S}^t$, then user $m\in\mathcal{M}^t_n$ in its coverage can offload data to SBS $n$ for edge computing. Otherwise, users in SBS $n$'s coverage have to offload data to the cloud (via the MBS) for cloud computing.

\subsection{ASP Utility Model}
The ASP derives utility by deploying edge computing services. On the one hand, the ASP has a larger utility if the edge computing service is deployed in areas where the service demand is larger as more users can enjoy a higher QoS. Let $d^t_m$ (in terms of the number of tasks) be the service demand of user $m$ in time slot $t$, which is unknown a priori at the edge service placement decision time, and the service demand of all users in the network is collected in $\d^t=(d^t_m)_{m\in\mathcal{M}^t}$. On the other hand, the ASP derives a larger utility if edge computing service is deployed in areas where edge computing performs much better than cloud computing. In this paper, we use delay as a performance metric of edge/cloud computing. Since we focus on a single application, we assume that tasks have same input data size $\lambda$ (in bits) and required CPU cycles $\eta$.

\subsubsection{Delay by Edge Computing}
If the task of user $m \in \mathcal{M}^t_n$ is processed by SBS $n$ at the edge, then the delay consists of the wireless transmission delay and edge computation delay. The achievable wireless uplink transmission rate between user $m$ and SBS $n$ can be calculated according to the Shannon capacity: $r^t_{n,m} = W\log_2\left(1 + P^{u}_m H^t_{n,m}/(N_0+I)\right)$, where $W$ is the channel bandwidth, $P^{u}_m$ is the transmission power of user $m$' device, $H^t_{n,m}$ is the uplink channel gain between user $m$ and SBS $n$ in time slot $t$, $N_0$ is the noise power and $I$ is the interference. Therefore, the transmission delay $Q^{\text{tx},t}_{n,m}$ of user $m$ for sending a task (i.e. $\lambda$ bits of input data) to SBS $n$ is $Q^{\text{tx},t}_{n,m} = \lambda/r^t_{n,m}$. We assume that the data size of the task result is small. Hence the downlink transmission delay is neglected. The computation delay depends on the computation workload and the edge server's CPU frequency. To simplify our analysis, we assume that the edge server of SBS $n$ processes tasks at its maximum CPU speed $f_n$. Therefore, the computation delay for one task is $\eta/f_n$. Overall, the delay of processing user $m$'s one task at the edge site of SBS $n$ is $Q^t_{n,m} = \lambda/r^t_{n,m} + \eta/f_n$.

\subsubsection{Delay by Cloud Computing}
If the task of user $m \in \mathcal{M}^t_n$ is processed in the cloud, then the delay consists of the wireless transmission delay, the backbone Internet transmission delay and the cloud computation delay. The wireless transmission delay can be computed similarly as in the edge computing case by first calculating the transmission rate $r^t_{0,m}$ between user $m$ and the MBS. The cloud computing delay can also be calculated similarly using the cloud server's CPU frequency $f_0$. However, compared to edge computing, cloud computing incurs an additional transmission delay since the data has to travel across the backbone Internet. Let $v^t$ be the backbone transmission rate and $h^t$ be the round trip time in time slot $t$, then an additional transmission delay $\lambda/v^t + h^t$ is incurred. Overall, the delay of processing user $m$'s one task in the cloud is $Q^{t}_{0,m} = \lambda/r^t_{0,m} + \eta/f_0 + (\lambda/v^t + h^t)$.

Taking into account the service demand and the possible delay reduction provided by edge computing. The utility of ASP when taking service placement decision $\mathcal{S}^t$ in time slot $t$ is:
\begin{align}\label{eq:utility}
U^t(\d^t, \mathcal{S}^t) = \sum_{n\in\mathcal{S}^t}\sum_{m\in\mathcal{M}^t_n}\tilde{u}^t_{n,m}d^t_m
\end{align}
where $\tilde{u}^t_{n,m} \triangleq Q^t_{0,m} - Q^t_{n,m}$ is the delay reduction for user $m$ if it is served by SBS $n$. The above utility function assumes that the tasks from a user are independent, i.e., the utility of a task is immediately realized upon the receival of its own results and does not need to wait until all tasks of the user are completed. Therefore, the ASP concerns the service delay for each individual tasks of users instead of measuring the delay for completing all the tasks of a user in time slot $t$. Similar utility functions are also widely adopted in the existing literature \cite{fan2018workload}. The utility is essentially a weighted service popularity, where the weight is the reduced delay by deploying edge services compared to cloud computing. Clearly, other weights, such as task/user priority, can also be easily incorporated into our framework. 

\textbf{Remarks on delay model:} We use simple models to capture the service delay incurred by task transmission and processing. Note that other communication models (e.g., Massive MIMO) and computing models (e.g., queuing system) can also be applied depending on the practical system configuration. In these cases, the delay reduction should be recalculated accordingly. The algorithm proposed in this paper is compatible with other delay models as long as tasks' delay reduction $\tilde{u}^t_{n,m}$ can be obtained.

\subsection{Context-Aware Edge Service Provisioning}
A user's service demand depends on many factors, which are collectively referred to as the \textit{context}. For example, relevant factors can be demographic factors (e.g., age\footnote{Young people are more interested in Game Apps as shown in \cite{lim2015investigating}.}, gender), equipment type (e.g. smartphone, tablet), equipment status (e.g., battery levels \footnote{A device with low battery level tend to offload computational tasks to edge servers \cite{mao2016dynamic}.} ), as well as external environment factors (e.g., location, time, and events). This categorization is clearly not exhaustive and the impact of each single context dimension on the service demand is unknown a priori. These context information helps the ASP to understand the demand pattern of connected users and provide the edge service efficiently. Our algorithm will learn to discover the underlying connection between such context and users' service demand pattern (see an example of such a connection in Figure \ref{fig:oracle_demand_est} based on a real-world dataset in Section \ref{sec:simulation}), which will be discussed in detail in the next subsection \ref{subsec:problem_formulation}, thereby facilitating the service placement decision making.

At each SBS, a context monitor periodically gathers context information by accessing information about currently connected users and optionally by collecting additional information from external sources (e.g. social media platforms). However, collecting the user context sometimes faces a concern known as the privacy disclosure management \cite{lederer2003managing}, which decides when, where, and what personal information can be revealed.  The central notion behind privacy disclosure management is that people disclose different versions of personal information to different entities under different conditions \cite{lederer2003managing}. Therefore, the service area of an SBS (e.g. business building, apartment complex, and plaza) may influence users' privacy preference \cite{anthony2007privacy,bilogrevic2016machine} and hence determine what context information an SBS can access. To capture this feature, we allow each SBS to maintain its own user context space depending on its local users' privacy preference. This results in the heterogeneity of context spaces maintained by SBSs. Note that the context spaces of different SBSs may be completely different, partially overlapping or exactly the same. Our model captures the most general cases and all SBSs having the same context space is a special case of ours. Formally, let $D_n$ be the number of context dimensions monitored by SBS $n$ for its connected users. The monitored context space of SBS $n$ is denoted by $\mathcal{X}_n$ which is assumed to be bounded and hence can be denoted as $\mathcal{X}_n=[0,1]^{D_n}$ without loss of generality. Let $x^t_{n,m}\in\mathcal{X}_n, \forall m\in \mathcal{M}^t_n$ be the context vector of user $m$ monitored by SBS $n$ in time slot $t$. The context vectors of all users connected to SBS $n$ are collected in $\x^t_n=(x^t_{n,m})_{m\in\mathcal{M}^t_n}$.

\subsection{Problem Formulation} \label{subsec:problem_formulation}
Now, we formulate the edge service placement problem as a CC-MAB learning problem. In each time slot $t$, the edge system operates sequentially as follows: (i) each SBS $n\in\mathcal{N}$ monitors the context $x^t_{n,m}\in\mathcal{X}_n$ of all connected users $m\in\mathcal{M}^t_n$ and collects the context information in $\x^t_n=(x^t_{n,m})_{m\in\mathcal{M}^t_n}$. (ii) The ASP chooses a set of SBSs $\mathcal{S}^t$ with $|\mathcal{S}^t| = b$ based on the context information collected by all SBSs $\X^t=(\x^t_n)_{n\in\mathcal{N}}$ in the current time slot, and the knowledge learned from previous time slots. (iii) The users are informed about the current service placement decision $\mathcal{S}^t$. Till the end of the current time slot, users connected to SBSs in $\mathcal{S}^t$ can request edge computing service from these SBSs. (iv) At the end of the current slot, the service demand $d^t_{m}$ of user $m\in\mathcal{M}^t_n$ served by SBS $n\in \mathcal{S}^t$ is observed.

The service demand $d^t_m$ of user $m\in\mathcal{M}^t_n$ with context $x^t_{n,m}\in\mathcal{X}_n$ is a random variable with a unknown distribution. We denote this random service demand by $d(x^t_{n,m})$ and its expected value by $\mu(x^t_{n,m}) = \mathbb{E}[d(x^t_{n,m})]$. The random service demand is assumed to take values in $[0,d^{\max}]$, where $d^{\max}$ is the maximum possible number of tasks a user can have in one time slot. The service demand $\left(d(x^t_{n,m})\right)_{n\in\mathcal{N},m\in\mathcal{M}^t_n}$ is assumed to be independent, i.e., the service demands of users served by an SBS are independent of each other. Moreover, each $d(x^t_{n,m})$ is assumed to be independent of the past service provision decisions and previous service demands.

The goal of the ASP is to rent at most $b$ SBSs for edge service hosting in order to maximize the expected utility up to a finite time horizon $T$. Based on the system utility defined in \eqref{eq:utility}, the edge service placement problem can be formally written as:
\begin{subequations}\label{eq:P1}
	\begin{align}\label{obj}
	\textbf{P1:}~~\max_{(\mathcal{S}^t)_{t=1,\dots,T}} & ~~\sum_{t=1}^{T}\sum_{n\in\mathcal{S}^t}\sum_{m \in \mathcal{M}^t_n} \tilde{u}^t_{n,m}\mu(x^t_{n,m})\\
	\text{s.t.} &~~|\mathcal{S}^t| \leq b, \mathcal{S}^t \subseteq \mathcal{N},~\forall t.
	\end{align}
\end{subequations}

\subsection{Oracle Benchmark Solution}
Before presenting our bandit learning algorithm, we first give an oracle benchmark solution to \textbf{P1} by assuming that the ASP had a priori knowledge about context-specific service demand, i.e., for an arbitrary user $m$ with context vector $x_{n,m}\in\mathcal{X}_n, \forall m\in\mathcal{M}_n, n\in\mathcal{N}$, the ASP would know the expected demand $\mu(x_{n,m})$. It is obvious that \textbf{P1} can be decoupled into $T$ independent sub-problems, one for each time slot $t$:
\begin{subequations}\label{eq:per_slot}
	\begin{align}
	\textbf{P2}~~\max_{\mathcal{S}^t} & ~~\sum_{n\in\mathcal{S}^t}\sum_{m \in \mathcal{M}^t_n} \tilde{u}^t_{n,m}\mu(x^t_{n,m})\\
	\text{s.t.} &~~|\mathcal{S}^t| \leq b, \mathcal{S}^t\subseteq\mathcal{N}.
	\end{align}
\end{subequations}
The optimal solution to the subproblem \textbf{P2} in time slot $t$ can be easily derived in a running time of $O(|N|\log(N))$ as follows: Given the contexts $\X^t$ of connected users in time slot $t$, the optimal solution is to select the $b$ highest ranked SBSs (top-$b$ SBSs) $\left\{n^*_1(\X^t), n^*_2(\X^t),\dots, n^*_b(\X^t)\right\} \in \mathcal{N}$ which, for $j=1,\dots,b$, satisfy:
\begin{align}\label{eq_oracle_solution}
n^*_j(\X^t)\in \argmax_{n\in\mathcal{N}\backslash\bigcup^{j-1}_{k=1}n^*_k(\X^t)}~\sum_{m \in \mathcal{M}^t_n} \tilde{u}^t_{n,m}\mu(x^t_{n,m})
\end{align}
We denote by $\mathcal{S}^{*t} = \bigcup^{b}_{k=1}n^*_k(\X^t)$ the optimal oracle solution in time slot $t$. Consequently, the collection $\left(\mathcal{S}^{*t}\right)_{t=1}^{T}$ is the optimal oracle solution to \textbf{P1}.

However, in practice, the ASP does not have a priori knowledge about the service demand. In this case, ASP cannot simply solve \textbf{P1} as described above, since the expected service demands are unknown. Hence, an ASP has to learn the expected service demand over time by observing the users' contexts and service demand. For this purpose, the ASP has to make a trade-off between deploying edge services at SBSs where little information is available (\emph{exploration}) and SBSs which it believes to yield the highest demands (\emph{exploitation}). In each time slot, the ASP's service placement decision depends on the history of choices in the past and observed user context. An algorithm which maps the decision history to the current service placement decision is called a learning algorithm. The oracle solution $\left(\mathcal{S}^{*t}\right)_{t=1}^{T}$, is used as a benchmark to evaluate the loss of learning. The regret of learning with respect to the oracle solution is given by
\begin{align}\label{eq:regret}
	R(T)=\sum_{t=1}^T \left( \sum_{n\in\mathcal{S}^{*t}} \sum_{m\in\mathcal{M}^t_n} \mathbb{E}\left[\tilde{u}^t_{n,m}d(x^t_{n,m})\right] - \sum_{n\in\mathcal{S}^t}\sum_{m\in\mathcal{M}^t_n} \mathbb{E}\left[\tilde{u}^t_{n,m}d(x^t_{n,m})\right] \right)
\end{align}
Here, the expectation is taken with respect to the decisions made by the learning algorithm and the distributions of users' service demand.

\section{CC-MAB for Edge Service Placement}\label{sec:SEEN}
In order to place edge services at the most beneficial SBSs given the context information of currently connected users, the ASP should learn context-specific service demand for the connected users. According to the above formulation, this problem is a combinatorial contextual MAB problem and we propose an algorithm called SEEN (Spatio-temporal Edge sErvice placemeNt) for learning the context-specific service demand and solving \textbf{P1}.

\subsection{Algorithm Structure}
Our SEEN algorithm is based on the assumption that users with similar context information covered by the same SBS will have similar service demand. This is a natural assumption in practice, which can be exploited together with the users' context information to learn future service provisioning decisions. Our algorithm starts by partitioning the context space maintained by each SBS uniformly into small hypercubes, i.e. splitting the entire context space into parts of similar contexts. Then, an SBS learns the service demand independently in each hypercube of similar contexts. Based on the observed context information of all connected users and a certain control function, the algorithm is interspersed with exploration phases and exploitation phases. In the exploration phases, ASP chooses a random set of SBSs for edge service placement. These phases are needed to learn the local users' service demand patterns of SBSs which have not been chosen often before. Otherwise, the algorithm is in an exploitation phase, in which it chooses SBSs which on average gave the highest utility when rented in previous time slots with similar user contexts. After choosing the new set of SBSs, the algorithm observes the users' true service demand at the end of every time slot. In this way, the algorithm learns context-specific service demand over time. The design challenge lies in how to partition the context space and how to determine when to explore/exploit.

The pseudo-code of SEEN is presented in Algorithm \ref{alg:SEEN}. In the initialization phase, SEEN creates a partition $\mathcal{P}_{n,T}$ for each SBS $n$ given the time horizon $T$, which splits the context space $\mathcal{X}_n=[0,1]^{D_n}$ into $(h_{n,T})^{D_n}$ sets and these sets are given by $D_n$-dimensional hypercubes of identical size $\dfrac{1}{h_{n,T}}\times \dots \times \dfrac{1}{h_{n,T}}$. Here, $h_{n,T}$ is an input parameter which determines the number of hypercubes in the partition. Additionally, SBS $n$ keeps a counter $C^t_{n}(p)$ for each hypercube $p\in\mathcal{P}_{n,T}$ indicating the number of times that a user with context from hypercube $p$ connects to SBS $n$ when it was rented to host edge service up to time slot $t$. Moreover, SEEN also keeps an estimated demand $\hat{d}(p)$ for each hypercube $p\in\mathcal{P}_{n,T}$. Let $\mathcal{E}^t_{n}(p)$ be the set of observed service demand of users with context from set $p\in\mathcal{P}_{n,T}$ up to time slot $t$. Then, the estimated demand of users with context from set $p$ is given by the sample mean:
\begin{align}
	\hat{d}(p)=\dfrac{1}{|\mathcal{E}^t_n(p)|}\sum_{d\in\mathcal{E}^t_n(p)}d
\end{align}
where $|\mathcal{E}^t_n(p)|$ equals $C^t_{n}(p)$. Notice that the set $\mathcal{E}^t_{n}(p)$ does not need be stored since the estimated demand $\hat{d}(p)$ can be updated based on $\mathcal{E}^{t-1}_{n}(p), C^{t-1}_{n}(p)$, and observed demands in time slot $t$.

In each time slot $t$, SBS $n$ first observes the currently connected users $\mathcal{M}^t_n$ and their context $\x_n=(x^t_{n,m})_{m\in\mathcal{M}^t_n}$. For each piece of context information $x^t_{n,m}$, SEEN determines the hypercube $p^t_{n,m}\in\mathcal{P}_{n,T}$ to which the $x^t_{n,m}$ belongs, i.e., $x^t_{n,m}\in p^t_{n,m}$ holds. The collection of these hypercubes is given by $\p^t_n=(p^t_{n,m})_{m\in\mathcal{M}^t_n}$ for each SBS $n\in\mathcal{N}$, and $\P^t=(\p^t_n)_{n\in\mathcal{N}}$ for the whole network. Fig. \ref{fig:context} offers a simple illustration of the context hypercubes and the update of counters with a 2-D context space assuming three users are currently connected to SBS $n$.
\begin{figure}
	\centering
	\includegraphics[width=0.6\linewidth]{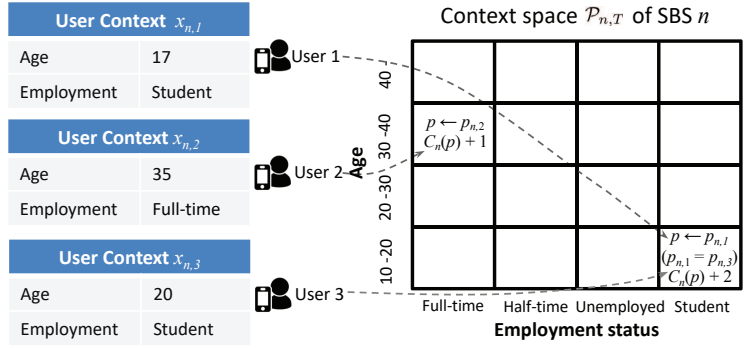}
	\caption{Illustration of context space and counters}
	\vspace{-0.3 in}
	\label{fig:context}
\end{figure}

Then the algorithm is in either an exploration phase or an exploitation phase. In order to determine the correct phase in the current time slot, the algorithm checks if there are SBSs that have not been explored sufficiently often. For this purpose, the set of \textit{under-explored} SBSs $\mathcal{N}^{\text{ue},t}$ are obtained in each time slot as follows:
\begin{align}\label{ue_SBS}
	\mathcal{N}^{\text{ue},t} =\{n:n\in\mathcal{N}, \exists~m\in \mathcal{M}^t_n, C^t_n(p^t_{n,m})<K_n(t)\}
\end{align}
where $K_n(t)$ is a deterministic, monotonically increasing control function, which is an input to the algorithm and has to be set appropriately to balance the trade-off between exploration and exploitation. In the next subsection, we will design a control function that guarantees a good balance in terms of this trade-off.


\begin{algorithm}[!t]
	\caption{Spatio-temporal Edge Service Provisioning (SEEN)} \label{alg:SEEN}
	\begin{algorithmic}[1]
		\State \textbf{Input}: $T$, $h_{n,T}$, $K_n(t)$.
		\State \textbf{Initialization} create partition $\mathcal{P}_{n,T}$; set $C^0_{n}(p)=0, \forall p \in \mathcal{P}_{n,T}, n\in\mathcal{N}$;
		\For {$t=1,\dots,T$}
		\State SBS $n\in\mathcal{N}$ observe currently connected users $\mathcal{M}^t_n$ and context $\x^t_n=(x^t_{n,m})_{m\in\mathcal{M}^t_n}$;
		\State Find $\p^t_n=(p^t_{n,m})_{m\in\mathcal{M}^t_n}$ such that $x^t_{n,m}\in p^t_{n,m}\in\mathcal{P}_{n,T}, \forall n\in\mathcal{N}, m\in\mathcal{M}^t_n$;
		\State Identify under-explored SBSs $\mathcal{N}^{\text{ue},t}$ in \eqref{ue_SBS} and let $q=\text{size}(\mathcal{N}^{\text{ue},t})$;
		\If {$\mathcal{N}^{\text{ue},t} \neq \emptyset$}: \Comment{\textit{Exploration}}
		\If{$q \geq b$}: $\mathcal{S}^t \leftarrow$ randomly rent $b$ SBSs from $\mathcal{N}^{\text{ue},t}$.
		\Else: $\mathcal{S}^t \leftarrow$ rent $q$ SBSs from $\mathcal{N}^{\text{ue},t}$, $(b-q)$ SBSs from $\left(\hat{n}_j(\X^t)\right)_{j=1}^{(b-q)}$ in \eqref{est_opt1};
		\EndIf
		\Else: $\mathcal{S}^t \leftarrow$ rent $b$ from $\left(\hat{n}_j(\X^t)\right)_{j=1}^{b}$ in \eqref{est_opt2}; \Comment{\textit{Exploitation}}
		\EndIf
		\State Observe service demand  $d_m$ of user $m$, $\forall m\in\mathcal{M}^t_n, \forall n\in\mathcal{S}^t$;
		\For {$n\in\mathcal{S}^t$ and $m\in\mathcal{M}^t_n $} \Comment{\textit{Demand estimation update}}
		\State Update estimated demand: $\hat{d}(p^t_{n,m})=\frac{\hat{d}(p^t_{n,m})C_n(p^t_{n,m})+d_{m}}{C_n(p^t_{n,m})+1}$; \label{line:demand_update}
		\State Update counters: $C_n(p^t_{n,m})=C_n(p^t_{n,m})+1$; \label{line:counter_update}
		\EndFor
		\EndFor
		\Statex \begin{footnotesize} The time indices of the counters $C^{t}_{n,p}$ are dropped in Line \ref{line:demand_update} and \ref{line:counter_update} due to recursive update. \end{footnotesize}	
	\end{algorithmic}
\end{algorithm}

If the set of under-explored SBSs is non-empty, SEEN enters the exploration phase. Let $q(t)$ be the number of under-explored SBSs. If the set of under-explored SBSs contains at least $b$ elements, i.e. $q(t)>b$, SEEN randomly rents $b$ SBSs from $\mathcal{N}^{\text{ue},t}$. If the number of under-explored SBS is less than $b$, i.e. $q(t)<b$, it selects $q(t)$ SBSs from $\mathcal{N}^{\text{ue},t}$ and $(b-q(t))$ additional SBSs are selected. These additional SBSs are those that have the highest estimated demand:
\begin{align}\label{est_opt1}
    \hat{n}^t_j(\X^t)\in \argmax_{n\in\mathcal{N}\backslash\{\mathcal{N}^{\text{ue},t},~\bigcup^{j-1}_{k=1}\hat{n}^t_k(\X^t)\}}~\sum_{m \in \mathcal{M}^t_n} \tilde{u}_{n,m}\hat{d}(p^t_{n,m}).
\end{align}
If the set of SBSs defined by \eqref{est_opt1} is not unique, ties are broken arbitrarily. If the set of under-explored SBSs is empty, then the algorithm enters the exploitation phase, in which it selects $b$ SBSs that have the highest estimated demand, as defined below:
\begin{align}\label{est_opt2}
\hat{n}^t_j(\X^t)\in \argmax_{n\in\mathcal{N}\backslash\bigcup^{j-1}_{k=1}\hat{n}^t_k(\X^t)}~\sum_{m \in \mathcal{M}^t_n}  \tilde{u}_{n,m}\hat{d}(p^t_{n,m}).
\end{align}

Finally, each chosen SBS observes the received service demand from users at the end of time slot $t$ and then updates the estimated service demand and the counters for each hypercube.

\subsection{Analysis of the Regret}
Next, we give an upper performance bound of the proposed algorithm in term of the \emph{regret}. The regret bound is derived based on the natural assumption that the expected service demands of users are similar in similar contexts. Because users' preferences of service demand differ based on their context, it is plausible for SBSs to divide its user population into groups with similar context and similar preferences. This assumption is formalized by the following H\"{o}lder condition for each SBS.
\begin{assumption}[H\"{o}lder Condition] \label{holder}
	For an arbitrary SBS $n\in\mathcal{N}$, there exists $L_n>0$, $\alpha_n>0$ such that for any $x,x^\prime\in\mathcal{X}_n$, it holds that
	\begin{align}
		|\mu(x)-\mu(x^\prime)| \leq L_n \|x-x^\prime\|^{\alpha_n}
	\end{align}
	where $\|\cdot\|$ denotes the Euclidean norm in $\mathbb{R}^{D_n}$.
\end{assumption}

We note that this assumption is needed for the analysis of the regret but SEEN can still be applied if it does not hold true. In that case, however, a regret bound might not be guaranteed. Under Assumption 1, the following Theorem shows that the regret of SEEN is sublinear in the time horizon $T$, i.e. $R(T)=O(T^{\gamma})$ with $\gamma<1$. This regret bound guarantees that SEEN has an asymptotically optimal performance, since $\lim_{T\to\infty} \frac{R(T)}{T}=0$ holds. This means that SEEN converges to the optimal edge service placement strategy used by the oracle solution. Specifically, the regret of SEEN can be bounded as follows for any finite time horizon $T$.

\begin{theorem}[Bound for $R(T)$] \label{theo:regret_bound}
	Let $K_n(t)=t^{\frac{2\alpha_n}{3\alpha_n+D_n}}\log(t)$ and $h_{n,T}=\lceil T^{\frac{1}{3\alpha_n+D_n}} \rceil$. If SEEN is run with these parameters and Assumption \ref{holder} holds true, the leading order of the regret $R(T)$ is $O\left(bN\tilde{u}^{\max}M^{\max}d^{\max} 2^{D_{\bar{n}}}T^{\frac{2\alpha_{\bar{n}}+D_{\bar{n}}}{3\alpha_{\bar{n}}+D_{\bar{n}}}}\log(T)\right)$, where $\bar{n}=\argmax_n \frac{2\alpha_n+D_n}{3\alpha_n+D_n}$.
\end{theorem}

\begin{proof}
	See online Appendix \ref{proof:theorem_bound_R(T)} in \cite{onlineappendix}.
\end{proof}
Theorem \ref{theo:regret_bound} indicates that the regret bound achieved by the proposed SEEN algorithm is sublinear in the time horizon $T$. Moreover, the bound is valid for any finite time horizon, thereby providing a bound on the performance loss for any finite number of service placement decision cycles. This can be used to characterize the convergence speed of the proposed algorithm. In the special case of $b = 1$ and $\alpha_{n} = \alpha, D_n = D, \forall n$, the considered CC-MAB problem reduces to the standard contextual MAB problem. In this case, the order of the regret is $\tilde{O}(T^\frac{2\alpha+D}{3\alpha+D})$. We note that the regret bound, which although is still sublinear in $T$, is loose when the budget $b$ is close to $N$. Consider the special case of $b = N$, SEEN actually is identical to the naive optimal service placement policy (i.e. choose all $N$ SBSs to deploy the edge service) and hence, the actual regret is 0. It is intuitive that when the budget $b$ is large, learning is not very much needed and hence the more challenging regime is when the budget $b$ is small (but not 1). 

\subsection{Complexity and Scalability}
The memory requirements of SEEN is mainly determined by the counters and estimated context-specific demands kept by the SBSs. For each SBS $n\in\mathcal{N}$, it keeps the counters $C_n(p)$ and estimated demand $\hat{d}(p)$ for each hypercube in the partition $\mathcal{P}_{n,T}$. If SEEN is run with the parameters in Theorem \ref{theo:regret_bound}, the number of hypercubes is $(h_{n,T})^{D_n} = \lceil T^{\frac{1}{3\alpha_n + D_n}}\rceil^{D_n}$. Hence, the required memory is sublinear in the time horizon $T$. However, this means that when $T \to \infty$, the algorithm would require infinite memory. Fortunately, in the practical implementations, SBS only needs to keep the counters of such hypercubes $p$ to which at least one of its connected users' context vectors belongs. Hence the required number of counters that have to be kept is actually much smaller than the analytical requirement.

SEEN can be easily implemented with a large network without incurring a large overhead, since each SBS keeps counters and estimated user demands independently according to its maintained context space. At the beginning of each time slot, the ASP queries the SBSs about their status (explored or under-explored) and estimated utilities, and then chooses $b$ SBSs based on SEEN. Therefore, the number of SBSs does not complicate the algorithm much.

\section{Edge Service Placement for SBSs with Coverage Overlapping} \label{sec:overlapping}
So far we have considered the edge service placement problem for a set of non-overlapping SBSs. However, SBSs may be densely deployed in areas with large mobile traffic data and computation demand, which creates the overlapping of SBSs' coverage. In this case, a user can be possibly served by multiple SBSs, and therefore whether a user's service demand can be processed at the Internet edge is determined by the service availability at all reachable SBSs. This creates spatial coupling of service demand among overlapped SBSs, i.e., users observed by an SBS may send service requests to other nearby SBSs. Therefore, it is difficult for the ASP to optimize the service placement policies by considering the service availability at each SBS separately. In this section, we propose an algorithm SEEN-O which extends SEEN for small-cell networks with coverage overlapping.

\subsection{SBS Component and Component-wise Service Provisioning}
We start by introducing the \emph{SBS components} and \emph{component-wise decision}. SEEN-O first constructs an undirected graph $G=\langle\mathcal{N},\mathcal{E}\rangle$ based on the small-cell network. Each SBS $n\in\mathcal{N}$ corresponds to a vertex in $G$. For each pair of vertices $i,j \in \mathcal{N}$, an edge $e_{i,j} \in \mathcal{E}$ is added between them if and only if the service areas of the two SBSs have coverage overlapping. Based on the constructed graph $G$, we give the definition of \emph{component} as follows:
\begin{definition}[Component]
	A component of an undirected graph is a subgraph in which any two vertices are connected to each other by paths, and which is connected to no additional vertices.
\end{definition}
By the definition of component, we know that a set of overlapped SBSs correspond to a component $C\subseteq\mathcal{N}$ in graph $G$. Let $\mathcal{C}=\{C_1,C_2,\dots,C_K\}$ collect all $K$ components in graph $G$. For an arbitrary component $C_k \in \mathcal{C}$, we define a set of component-wise decisions $Z_k = \{z : z \subseteq C_k, z \neq \emptyset\}$, which collects all possible service placement decisions for SBSs in component $C_k$. The component-wise decision set $Z_k$ can also be written as $Z_k = \{z_{k,1},z_{k,2},\dots,z_{k,L}\}$, where the total number of decisions in $Z_k$ is given by the Bell number $L=\sum_{j=1}^{|C_k|} {|C_k| \choose j}$. For an arbitrary component $C_k$, if a component-wise decision $z_{k,l}\in Z_k$ is taken, then the ASP rents SBSs $ n\in z_{k,l}$ from the set of overlapping SBSs in $C_i$. Notice that the non-overlapping SBS network is a special case: for the components $C_k$ containing only one SBS $n$ (i.e., non-overlapping SBS), its component-wise decision set $Z_k$ contains only one element $z_{k,1}=\{n\}$. Let $\mathcal{Z}=\bigcup^K_{k=1} Z_k$ be component-wise decision sets for the whole network. Fig. \ref{fig:component} provides a simple illustration of the SBS components and component-wise decisions.
\begin{figure}[htb]
	\centering
	\includegraphics[width=0.8\linewidth]{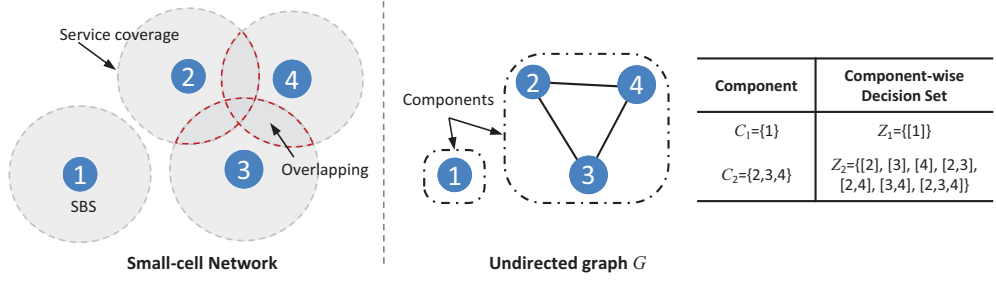}
	\caption{Illustration of SBS components and component-wise decisions. The component-wise decision set for the whole network is $\mathcal{Z} = \left\{[1], [2], [3], [4], [2,3], [2,4], [3,4], [2,3,4]\right\}$.}
	\label{fig:component}
	\vspace{-0.2 in}
\end{figure}

Instead of picking service placement decisions for individual SBSs separately, SEEN-O chooses component-wise decisions for components $C_k\in\mathcal{C}$ due to the fact that service demand received by an SBS is jointly decided by the service availability at SBSs in the same component. Let $\mathcal{M}_{C_k}$ denote the users collaboratively served by SBSs in component $C_k$. For a user $m\in \mathcal{M}_{C_k}$, it is able to request edge services from multiple SBSs in the component depending on the chosen component-wise decision $z_{k,l}\in Z_k$. Let $H^t_{n,m}$ be the uplink channel gain between user $m$ and SBS $n \in C_k$. If SBS $n$ is not reachable for user $m$, then $H^t_{n,m} = 0$. Usually, users' devices are energy-constrained and hence we assume that the service demand $d^t_m$ of user $m\in\mathcal{M}_{C_i}$ is offloaded to the SBS that has the best uplink channel condition among those that can provide edge service, namely $\arg\max_{n\in z_{k,l}} H^t_{n,m}$. In this way, users incur the least transmission energy consumption \footnote{Our algorithm is also compatible with other User-SBS association strategies}. The association decision $a_m$ of user $m$ can be formally written as:
\begin{align}
	a_m(z) = \argmax_{n \in C_k}~H^t_{n,m} \cdot\textbf{1}\{n \in z\}, m\in\mathcal{M}_{C_k}, z\in Z_k.
\end{align}
Note that the uplink channel conditions can be easily monitored by the users, and we also assume that the users report monitored channel conditions to all reachable SBSs. Therefore, the association decisions of user $m\in\mathcal{M}_{C_k}$ are known to the SBSs given the component-wise decision $z$. Let $\mathcal{M}_n(C_k,z)$ be the users connected to SBS $n \in z, z \subseteq C_k$, we have:
\begin{align}
	\mathcal{M}_n(C_k,z) = \{m : m \in \mathcal{M}_{C_k}, a_m(z) = n\}, n\in C_k, z\in Z_k.
\end{align}
In addition, for each SBS $n$ we define $\tilde{u}_n(z) = \sum_{m\in\mathcal{M}_n(C_k,z)} \tilde{u}_{n,m}\mu(x^t_{n,m})$, where $\tilde{u}_{n,m}$ is the delay improvement of user $m$ as defined in \eqref{eq:utility}. Let $\mathcal{S}^t_z \subseteq \mathcal{Z}$ be the component-wise decisions chosen by the ASP. Notice that the ASP can only draw at most one component-wise decision $z\in Z_k$ for each component $C_k$. Then, we have the edge service placement problem as follows:
\begin{subequations}\label{eq:P2}
	\begin{align}
	\textbf{P3:}~~~\max_{(\mathcal{S}_z^t)_{t=1,\dots,T}} &  ~~\sum_{t=1}^{T}\sum_{z\in\mathcal{S}_z^t} \sum_{n \in z} \tilde{u}_n (z) \\
	\text{s.t.} ~~ &\sum_{z\in\mathcal{S}^t_z} |z| \leq b,~~ \forall~t \label{eq:P2budget}\\
	&\sum_{z\in\mathcal{S}^t_z} {\textbf{1}\{z \in Z_k\}} \leq 1, ~~\forall~k,\forall~t \label{eq:conflict}
	\end{align}
\end{subequations}
where \eqref{eq:P2budget} is the budget constraint for the ASP and \eqref{eq:conflict} indicates that only one component-wise decision can be selected for each component.

\subsection{Disjunctively Constrained Knapsack Problem}
Now, we consider an oracle solution for \textbf{P3}. Similarly, \textbf{P3} can be decoupled into $T$ sub-problems.Yet, the solution for each sub-problem cannot be easily derived as in \eqref{eq_oracle_solution} due to different costs incurred by different component-wise decisions and, more importantly, the conflicts among component-wise decisions in \eqref{eq:conflict}. The per-slot subproblem of \textbf{P3} can be formulated as a \emph{Knapsack problem with Conflict Graphs} (KCG), which is also referred to as \emph{disjunctively constrained knapsack problem}. The conflict graph $G_c=\langle\mathcal{Z},\mathcal{E}_c\rangle$ is defined based on the component-wise decisions: Each component-wise decision $z\in\mathcal{Z}$ corresponds to a vertex in $G_c$. For an arbitrary pair of vertices $z,z^\prime\in\mathcal{Z}$, add an edge $e(z,z^\prime)\in\mathcal{E}_c$ between $z$ and $z^\prime$ if there exist a component-wise decision set $Z_k\in\mathcal{Z}$ such that $z,z^\prime \in Z_k$.

In the following, we convert \textbf{P3} to a standard formulation of the KCG problem. For each $z\in\mathcal{Z}$, we define a tuple $(p_z,c_z,y_z)$ where $p_z= \sum_{n \in z} \tilde{u}_n (z)$ is the profit of choosing component-wise decision $z$, $c_z$ is the cost of decision $z$ which equals $|z|$, and $y_z\in\{0,1\}$ indicates whether the decision $z$ is taken or not. Then, a KCG problem equivalent to \textbf{P3} can be written as:
\begin{subequations}
	\begin{align}
     	\textbf{P3-KCG}:~~~& \max~\sum_{z\in\mathcal{Z}}  p_zy_z\\
     	\text{s.t.} ~~ & \sum_{z \in \mathcal{Z}} c_zy_z\leq b\\
     		   & y_{z}+ y_{z^\prime} < 1, \forall z,z^\prime \in \mathcal{Z}, \forall e(z,z^\prime)\in\mathcal{E}_c\\
     		   & y_z \in \{0,1\}, \forall z\in\mathcal{Z}
	\end{align}
\end{subequations}

The above problem is an NP-hard combinatorial optimization problem. Existing works have proposed various algorithms, including heuristic solutions\cite{yamada2002heuristic} and exact solutions \cite{hifi2012algorithm} to solve KCG. In the simulation we employ the \emph{Branch-and-Bound} algorithm \cite{bettinelli2017branch} to solve \textbf{P3-KCG}.

\begin{algorithm}[htb]
	\caption{SEEN-O} \label{alg:SEEN_overlap}
	\begin{algorithmic}[1]
		\State \textbf{Input}: $T$, $h_{n,T}$, $K_n(t)$.
		\State \textbf{Initialization} context partition: $\mathcal{P}_{n,T}$; set $C^0_{n}(p)=0, \forall p \in \mathcal{P}_{n,T}, n\in\mathcal{N}$;
		\For {$t=1,\dots,T$}
		\State SBS $n\in\mathcal{N}$ observe currently connected users $\mathcal{M}^t_n$ and context $\x^t_n=(x^t_{n,m})_{m\in\mathcal{M}^t_n}$;
		\State Find $\p^t_n=(p^t_{n,m})_{m\in\mathcal{M}^t_n}$ such that $x^t_{n,m}\in p^t_{n,m}\in\mathcal{P}_{n,T}, \forall n\in\mathcal{N}, m\in\mathcal{M}^t_n$;
		\State Identify under-explored SBSs $\mathcal{N}^{\text{ue},t}$ using \eqref{ue_SBS} and let $q=\text{size}(\mathcal{N}^{\text{ue},t})$;
		\If {$\mathcal{N}^{\text{ue},t} \neq \emptyset$}: \Comment{\textit{Exploration}}
		\If{$q \geq b$}: ASP randomly rents $b$ SBSs from $\mathcal{N}^{\text{ue},t}$;
		\Else: ASP rents $q$ SBSs from $\mathcal{N}^{\text{ue},t}$;
		\State \quad Identify $\tilde{\mathcal{Z}}$ and select other $(b-q)$ SBSs by solving the KCG with $\tilde{\mathcal{Z}}, \tilde{c}_z, \tilde{p}_z, \tilde{b}$;
		\EndIf
		\Else:  Select $b$ SBSs by solving \textbf{P3-KCG} with current demand estimation; \Comment{\textit{Exploitation}}
		\EndIf
		\For {each user $m$ served by SBSs}: \Comment{\textit{Demand estimation update}}
		\For {each SBS that covers user $m$, i.e., $n\in\{n:m\in\mathcal{M}_n, n\in\mathcal{N}\}$}
		\State update demand estimation: $\hat{d}(p^t_{n,m})=\frac{\hat{d}(p^t_{n,m})C_n(p^t_{n,m})+d_{m}}{C_n(p^t_{n,m})+1}$; \label{line:demand_update_o}
		\State update counters: $C_n(p^t_{n,m})=C_n(p^t_{n,m})+1$; \label{line:counter_update_o}
		\EndFor
		\EndFor
		\EndFor
	\end{algorithmic}
\end{algorithm}

\subsection{Algorithm Structure}
Now, we give SEEN-O in Algorithm \ref{alg:SEEN_overlap} for edge service placement with coverage overlapping. Similar to SEEN, SEEN-O also has two phases: exploration and exploitation. We first obtain the set of \textit{under-explored} SBSs $\mathcal{N}^{\text{ue},t}$ as in \eqref{ue_SBS} based on users in their coverage $\mathcal{M}_n$ \footnote{$\mathcal{M}_n$ is the set of users within the coverage of SBS $n$. Note that it is different from $\mathcal{M}_n(C_k, z)$ which denotes the users served by SBS $n$ depending on the component-wise decisions.}. If the set of under-explored SBSs is non-empty, namely $\mathcal{N}^{\text{ue},t} \neq \emptyset$, then SEEN-O enters the exploration phase. Let $q$ be the number of under-explored SBSs. If the set of under-explored SBSs contains at least $b$ elements $q>b$, SEEN-O randomly rents $b$ SBSs from $\mathcal{N}^{\text{ue},t}$. If the number of under-explored SBS is less than $b$, i.e. $q<b$, SEEN-O first selects $q$ SBSs from $\mathcal{N}^{\text{ue},t}$ and $(b-q)$ SBSs are selected by solving a KCG problem based on the following component-wise decisions:
\begin{align} \label{tilde_Z}
&\tilde{Z}_k = Z_k \backslash \left\{Z_k^{\prime} \cup Z_k^{\prime\prime}\right\}\\
\text{where}~~ Z_k^{\prime} = \{z :& z \in Z_k,~\text{if}~\exists~n \in \mathcal{N}^{\text{ue},t}, \{n\} = z \}\\
Z_k^{\prime\prime} = \{z :& z \in Z_k,~\text{if}~\exists~n \in \mathcal{N}^{\text{ue},t} \cap C_k, n \notin z \}
\end{align}
$Z_k^{\prime}$ is the set of one-element component-wise decision $\{n\}, \forall n \in \mathcal{N}^{\text{ue},t}$. The decisions in $Z_k^{\prime}$ need to be removed from $Z_k$ since they have already been chosen by ASP; $Z_k^{\prime\prime}$ collects component-wise decisions for $C_k$ which do not contain the under-explored SBSs $n \in \{\mathcal{N}^{\text{ue},t} \cap C_k\}$. The decisions in $Z_k^{\prime\prime}$ are also removed since the component-wise decision for $C_k$ must contain all the under-explored SBSs in $C_k$. Then, ASP solves a KCG problem with the constructed component-wise decision set $\tilde{\mathcal{Z}}=\{\tilde{Z}_1,\dots,\tilde{Z}_K\}$, decision $\tilde{c}_z = |z\backslash \mathcal{N}^{\text{ue},t}|$, decision profit $\tilde{p}_z= \sum_{n \in z\backslash \mathcal{N}^{\text{ue},t}} \tilde{u}_n (z)$, and the modified budget $\tilde{b} = b-q$. If the set of under-explored SBSs is empty, then the algorithm enters the exploitation phase. It solves the \textbf{P3-KCG} based on the current context-specific demand estimation with all component-wise decision $\mathcal{Z}$ and budget $b$.

At the end of each time slot, SBSs observe service demand received from the connected users. Then, each SBS updates the estimated service demand and the counters for each context hypercube. Notice that in the overlapping case, a user can be covered by multiple SBSs and therefore, the observed service demand can be used to update the estimated service demand at multiple SBSs. For example, if a user is in the coverage of SBS $i$ and SBS $j$, namely $m\in \mathcal{M}_i\cap\mathcal{M}_j$. Then, the observed service demand of this user can be used to update the context-specific service demand estimation at both SBS $i$ and SBS $j$. This also means that SEEN-O can learn the reward of multiple component-wise decisions in one time slot, e.g. if component-wise decision $[i,j]$ is taken. The utility of component-wise decisions $[i],[j],[i,j]$ can be updated at the same time. Theorem 2 proves that SEEN-O has the same regret bound as SEEN.
\begin{theorem}[Regret Bound for SEEN-O] \label{theo:regret_bound_seen_O}
SEEN-O has the same regret bound as SEEN.
\end{theorem}
\begin{proof}
	See online Appendix \ref{proof:theorem_bound_seen_o} in \cite{onlineappendix}.
\end{proof}
The regret upper bound for SEEN-O in Theorem 2 is valid for any edge network layout and does not require any assumption on SBS deployment and user population distribution. This helps to carry out SEEN-O in a practical application since, in most cases, the SBS deployment is revealed to ASP though, the user distribution is unknown a priori.

\section{Simulation}\label{sec:simulation}
In this section, we carry out simulations on a real-world dataset to evaluate the performance of the proposed algorithms. We use the data collected in \cite{lim2015investigating} which aims to reveal the underlying link between the demand for mobile applications and the user context including age, gender, occupation, years of education, device type (e.g. phone, tablet, and laptop), and nationality. It collects the context information of a total of 10,208 end users and the users' demand for 23 types of mobile applications. We envision that these mobile applications can be deployed on edge servers at SBSs via containerization and the UEs can send computing tasks to SBS for processing. In our simulation, we consider that the ASP aims to provide edge service for \emph{Game-type} application (the most popular application out of 23 mobile applications investigated in \cite{lim2015investigating}), which is also a major use case of edge computing. Fig. \ref{fig:userdist_demand} and Fig \ref{fig:userdist_nodemand} depict the user distribution, and Fig. \ref{fig:oracle_demand_est} depicts the context-specific service demand estimation on the two context dimensions \emph{Age} and \emph{Years of education}. We see clearly that the users' demand pattern is very related to the users' context information. Note that the \emph{Age} and \emph{Years of education} information is obtained from the dataset \cite{lim2015investigating} and is used only as an example to illustrate the context-demand relationship. In practice, users may be willing to disclose such information in enterprise or campus internal networks. For the more general scenario, SBSs can use other less sensitive context such as user device information.

\begin{figure}[htb]
	\subfigure[Users with service demand]{\label{fig:userdist_demand}
		\includegraphics[width=0.31\textwidth]{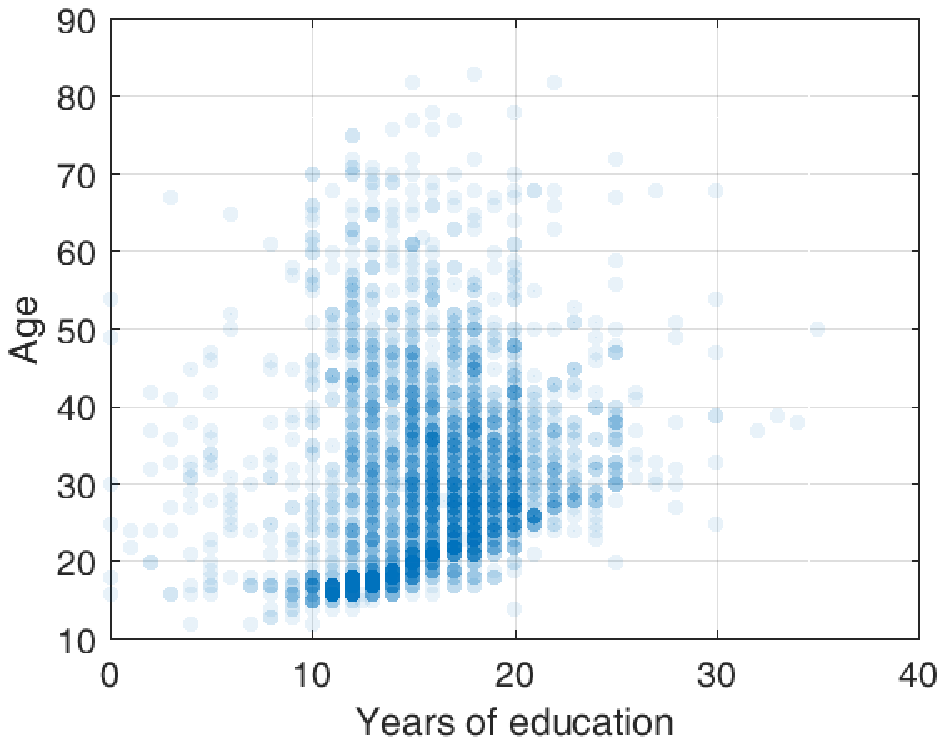}}
	\subfigure[Users with no service demand]{\label{fig:userdist_nodemand}
		\includegraphics[width=0.31\textwidth]{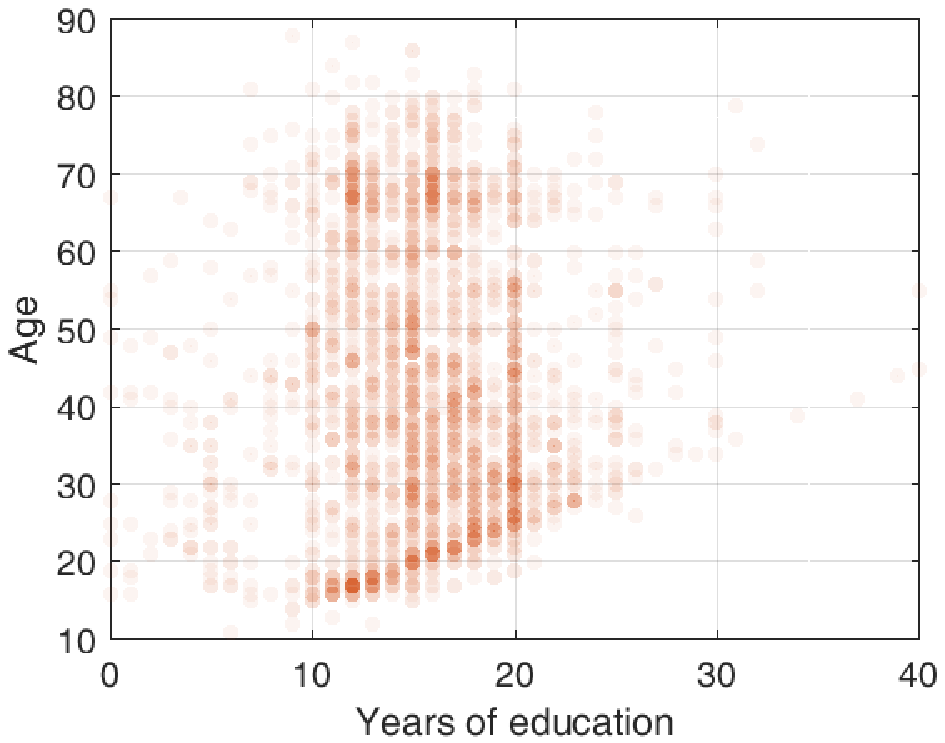}}
	\subfigure[Oracle demand estimation]{\label{fig:oracle_demand_est}
		\includegraphics[width=0.34\textwidth]{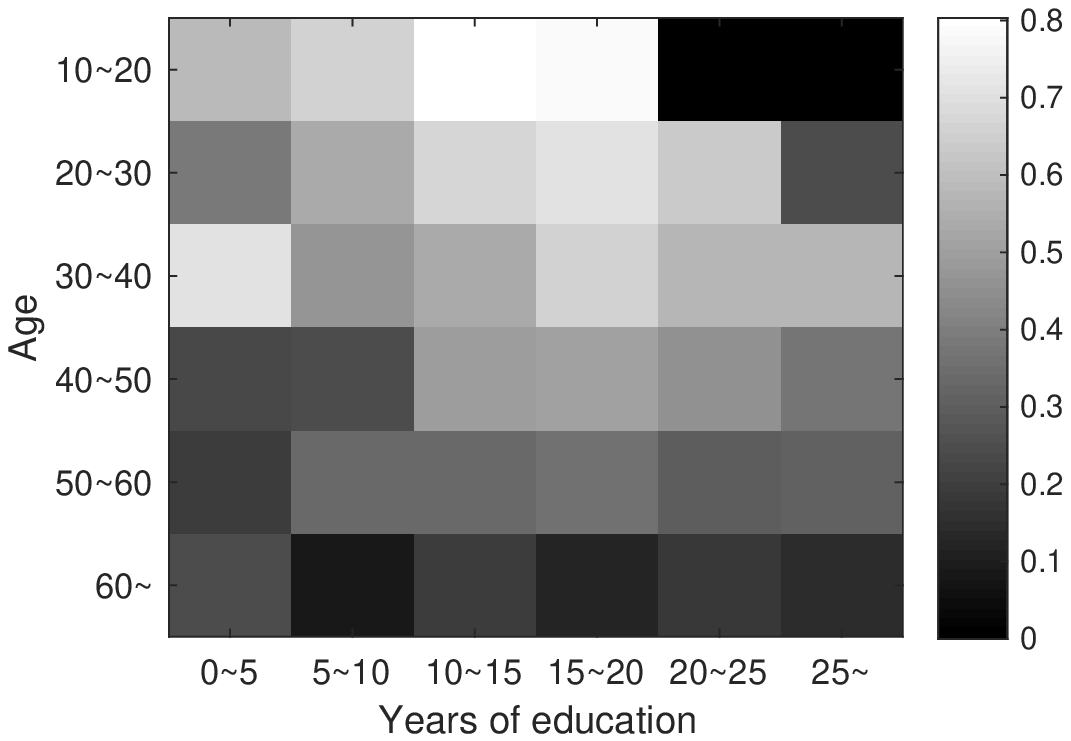}}
	
	\caption{User demand pattern on \emph{age} and \emph{year of education} dimension} \label{fig:userdist}
	\vspace{-0.25 in}
\end{figure}


For the small-cell network, we simulate a 1000m$\times$1000m area served by $N=10$ SBSs and one MBS. The SBSs are randomly scattered in this area. An SBS can serve users within the service range 150m, which tends to create coverage overlapping among SBSs. For the analysis of non-overlapping SBSs, we assume that users request edge service only from the nearest SBSs; while, in the overlapping case, a user is allowed to decide its association based on the service availability and channel condition of reachable SBSs. To capture different compositions of user population across different SBSs, we randomly assign one out of three area types (\emph{school zone, business area, and public}) to each SBS, where users with \emph{student} occupation context tend to show up in school zones with a higher probability, users with \emph{full-time worker} tend to show up in business areas, and all types of users show up in public with the same probability. The default value of ASP budget is set as $b=3$. Other key simulation parameters are: channel bandwidth $W=20$MHz, transmission power of user equipment $P^{u}_m=10$dBm, noise power $\sigma^2=10^{-10}$W/Hz, CPU frequency at SBSs $f_n=2.8$GHz, CPU frequency at the cloud $f_0 = 5.6$GHz, Internet backhaul transmission rate $v^t \in [10,20]$Mbps, round-trip time $h = 100$ms.

The proposed algorithm is compared with the following benchmarks:\\
(1) Oracle algorithm: Oracle knows precisely the expected demand for any user context. In each time slot, Oracle selects $b$ SBSs that maximize the expected system utility as in \eqref{eq_oracle_solution} based on the observed user context.\\
(2) Combinatorial UCB (cUCB)\cite{chen2013combinatorial}: cUCB is developed based on a classic MAB algorithm, UCB1. The key idea is to create super-arms, i.e., $b$-element combination of $N$ SBS ($b$ is the budget). There will be a total of $N \choose b$ super-arms and cUCB learns the reward of each super-arm.\\
(3) Combinatorial-Contextual UCB (c$^2$UCB): c$^2$UCB considers users' context when running cUCB. Specifically, c$^2$UCB maintains a context space for each super-arm and the utility estimations of hypercubes in a context space are updated when corresponding super-arm is selected.\\
(4) $\epsilon$-Greedy: $\epsilon$-Greedy rents a random set of $b$ SBSs with probability $\epsilon\in(0,1)$. With probability $(1-\epsilon)$, the algorithm selects $b$ SBSs with highest estimated demands. These estimated demands are calculated based on the previously observed demand of rented SBSs.\\
(5) Random algorithm: The algorithm simply rents $b$ SBSs randomly in each time slot.

\subsection{Performance Comparison}
Fig. \ref{fig:cum_sys_utility} shows the cumulative system utility achieved by SEEN and other 5 benchmarks for a non-overlapping case. As expected, the Oracle algorithm has the highest cumulative system utility and gives an upper bound to the other algorithms. Among the other algorithms, we see that SEEN and c$^2$UCB significantly outperform cUCB, $\epsilon$-Greedy, and Random algorithm, since they take into account the context information when estimating the users' service demand pattern. Moreover, SEEN achieves a higher system utility compared with c$^2$UCB. This is due to the fact that c$^2$UCB creates a large set of super-arms and therefore is more likely to enter the exploration phase. The conventional algorithms, cUCB, and $\epsilon$-Greedy, only provide slight improvements compared to the Random algorithm. The failure of these methods is due to the uncertainty of user population in various aspects, e.g. user numbers and composition, which are difficult to estimate in each time slot without observing the user context information.
%
\begin{figure*}[htb]
	\begin{minipage}[t]{0.5\linewidth}
		\centering
		\includegraphics[width=0.95\linewidth]{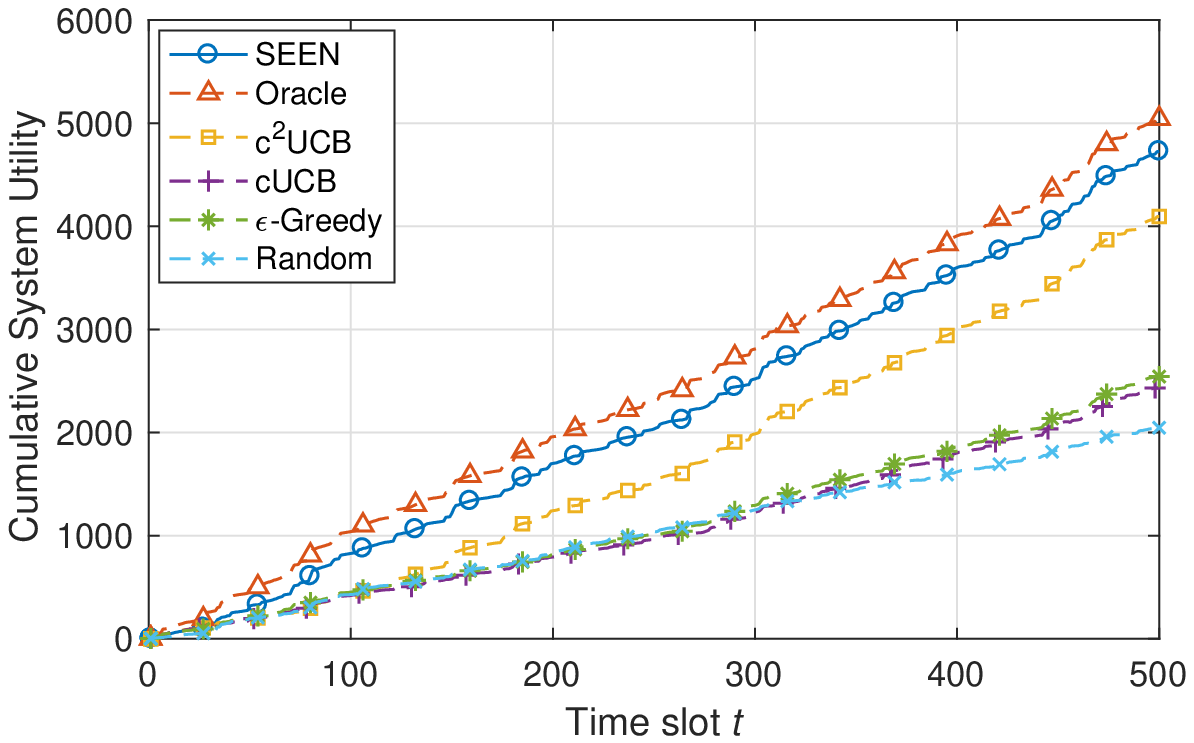}
		\vspace{-0.1 in}
		\caption{Comparison on cumulative system utility.}
		\label{fig:cum_sys_utility}
	\end{minipage}%
	\begin{minipage}[t]{0.5\linewidth}
		\centering
		\includegraphics[width=0.95\linewidth]{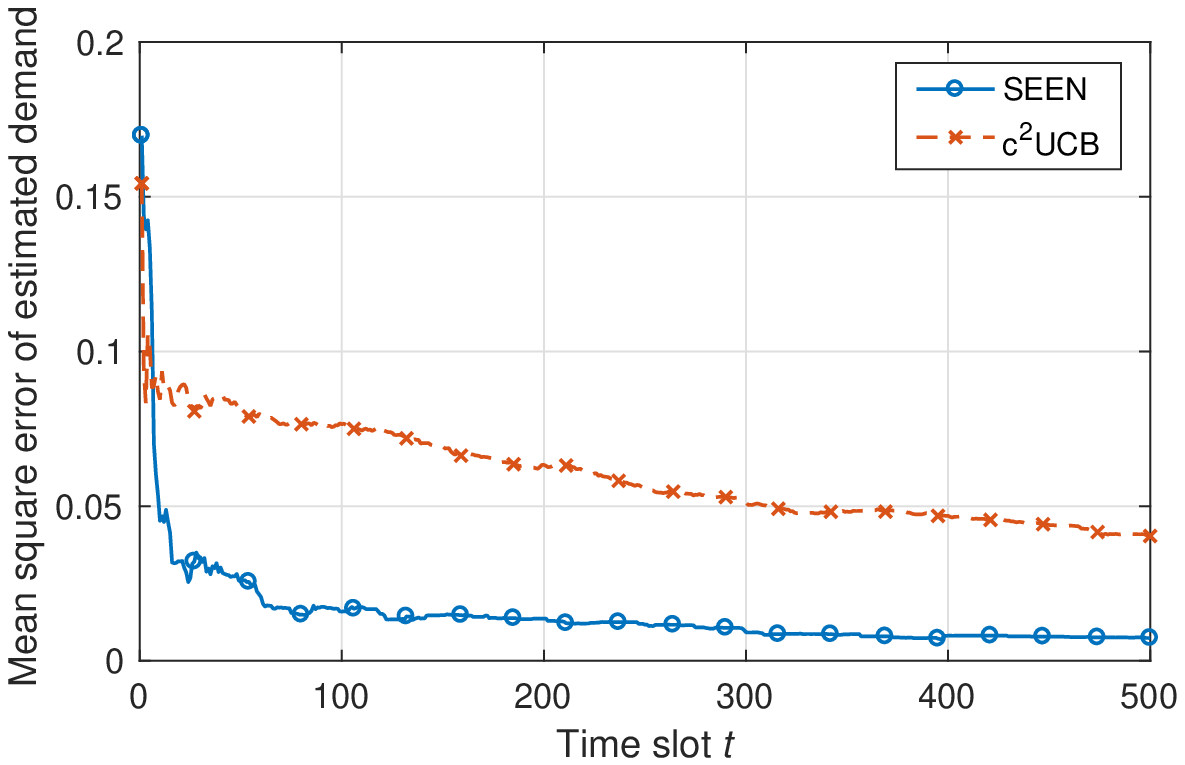}
		\vspace{-0.1 in}
		\caption{MSE of estimated service demand}
		\label{fig:MES_est_demand}
	\end{minipage}%
	\vspace{-0.1 in}
\end{figure*}

\subsection{Demand estimation error}
Fig. \ref{fig:MES_est_demand} shows the mean square error (MSE) of service demand estimation achieved by SEEN and c$^2$UCB, where the MSE is measured across all context hypercubes compared to the oracle demand estimation. It can be observed that the MSE of SEEN converges quickly to 0.01 after first 120 time slots while the MSE of c$^2$UCB stays high and decreases slowly during 500-slot runtime. This means that SEEN is able to learn the user demand pattern fast and provide more effective decisions on edge service placement.

\subsection{Demand allocation}
Fig. \ref{fig:demand_allocation} shows the allocation of user demand in the network, i.e., whether the demand is processed at the edge or cloud. Note that ASP desires to process more demand at the edge so that lower delay costs are incurred to users. We can see from Fig. \ref{fig:demand_allocation} that SEEN is able to accommodate a large amount of user demand 62.2\%, which is slightly lower than that of Oracle (69.2\%). For other four schemes, they rely heavily on the cloud server, therefore incurring large delay cost and diminishing the system utility.
\begin{figure}[htb]
	\centering
	\includegraphics[width=0.5\linewidth]{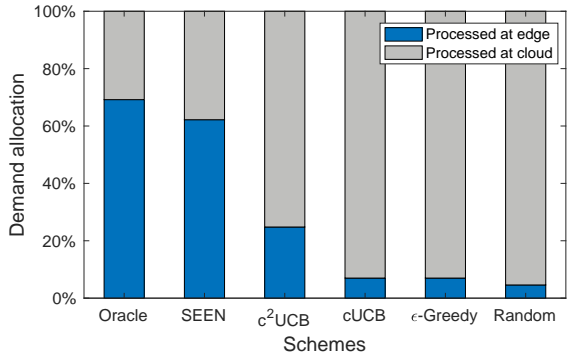}
	\caption{User demand allocation}
	\label{fig:demand_allocation}
	\vspace{-0.3 in}
\end{figure}

\subsection{Learning with More Context}
Next, we evaluate the performance of SEEN under different context spaces. Figure \ref{fig:context_num} shows the cumulative system utilities achieved by SEEN and 5 other benchmarks when running with 2, 3, 4 contexts. Comparing these three figures, we see that the cumulative system utilities achieved by cUCB, $\epsilon$-Greedy, and Random stay more or less the same, since these algorithms are independent of the context information. The context-aware algorithms, i.e., SEEN, Oracle, and c$^2$UCB, achieve higher cumulative utilities with more context information since more contexts help the ASP to learn the users' demand pattern and therefore make better service provisioning decision. In addition, it is worth noticing that SEEN incurs larger regrets when running with more context information, which is consistent with the analysis in Theorem \ref{theo:regret_bound}. 
\begin{figure}[htb]
	\centering	
	\subfigure[\emph{age}, \emph{employment status}]{\label{fig:cum_sys_utility_cn2}
		\includegraphics[width=0.45\linewidth]{Figure/cum_sys_utility.eps}}
	\subfigure[\emph{age}, \emph{employment status}, \emph{marital status}]{\label{fig:cum_sys_utility_cn3}
		\includegraphics[width=0.45\linewidth]{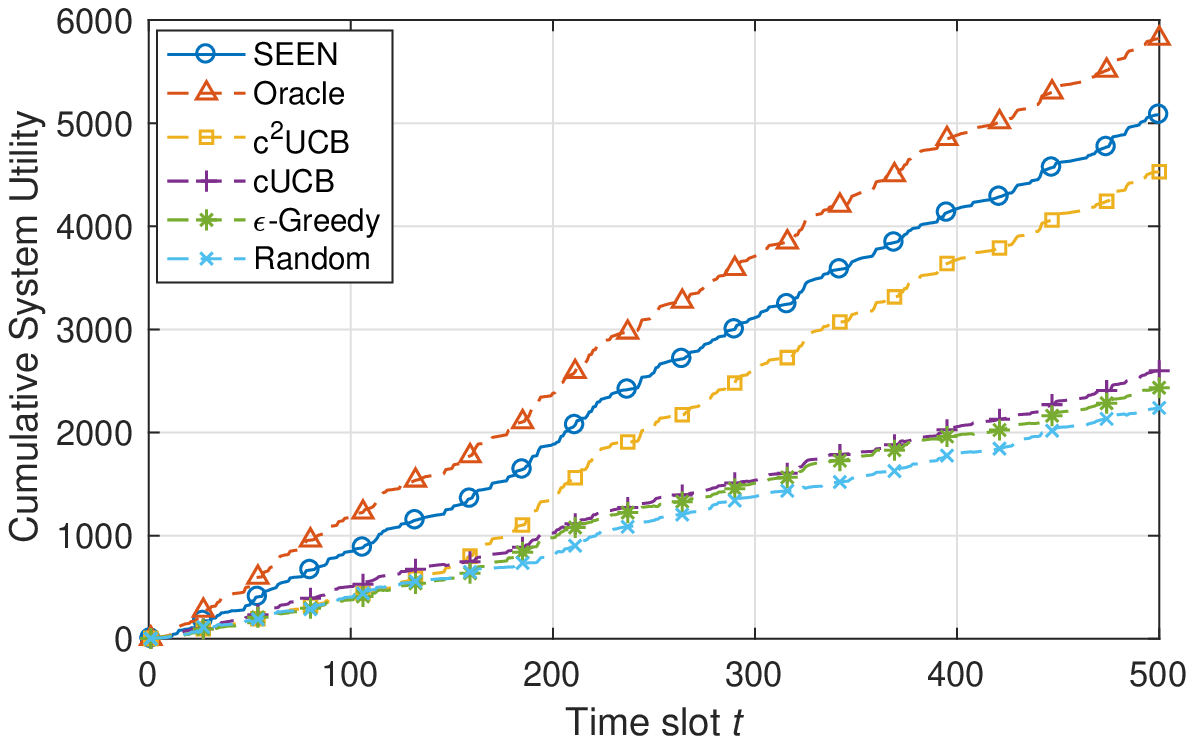}}
	\subfigure[\emph{age}, \emph{download purpose}, \emph{employment and marital status}]{\label{fig:cum_sys_utility_cn4}
		\includegraphics[width=0.45\linewidth]{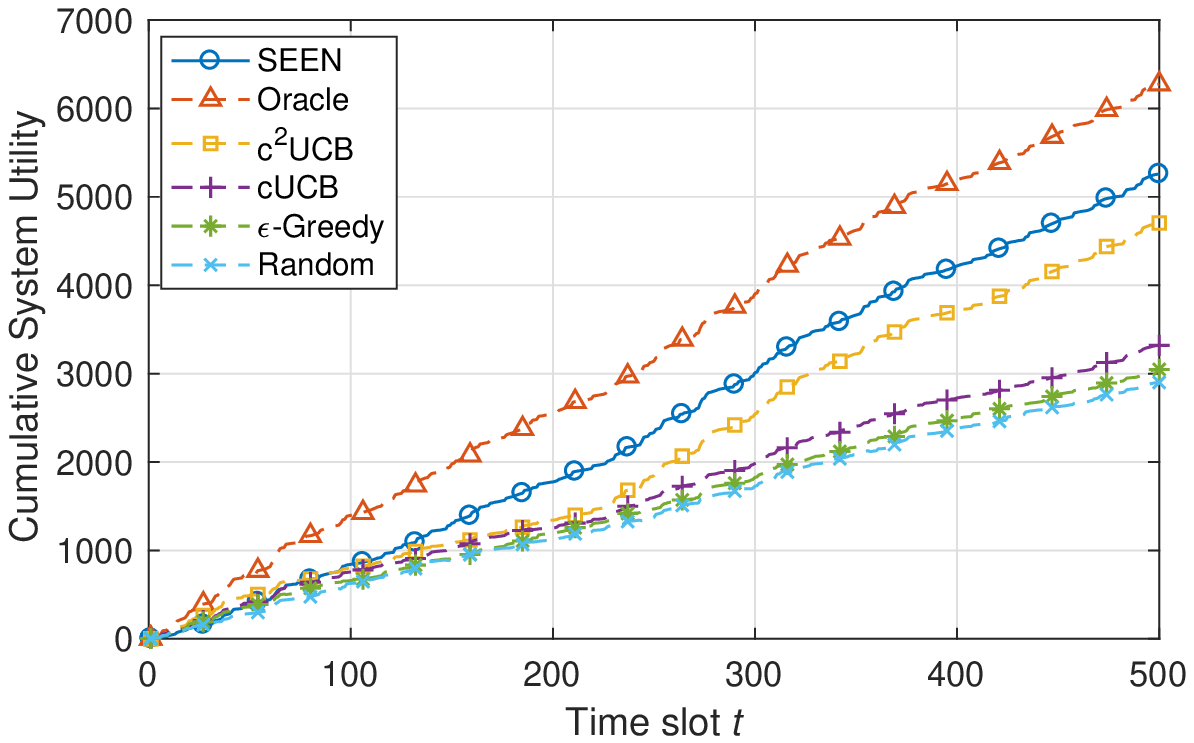}}
	\caption{Learning with different context spaces. (a) two-dimension (b) three-dimension (c) four-dimension}
	\label{fig:context_num}
	\vspace{-0.3 in}
\end{figure}

\subsection{Impact of ASP budget}
Fig. \ref{fig:varyb} depicts the cumulative system utility achieved by 6 schemes in 500 slots with different budgets. As expected, the system utility grows with the increase in ASP budget $b$ since more user demand can be processed at the network edge with more SBSs providing edge services. Moreover, we see that SEEN is able to achieve close-to-oracle performance at all levels of ASP budget. By contrast, the c$^2$UCB algorithm suffers an obvious performance degradation with $b\in [4,6]$. This is due to the fact that number of super-arms $N \choose b$ created by c$^2$UCB becomes very large given $N=10$ and $b\in[4,6]$. This forces the c$^2$UCB algorithm to enter exploration more frequently and leads to system utility loss.

\begin{figure*}[t]
	\begin{minipage}[t]{0.5\linewidth}
		\centering
		\includegraphics[width=0.95\linewidth]{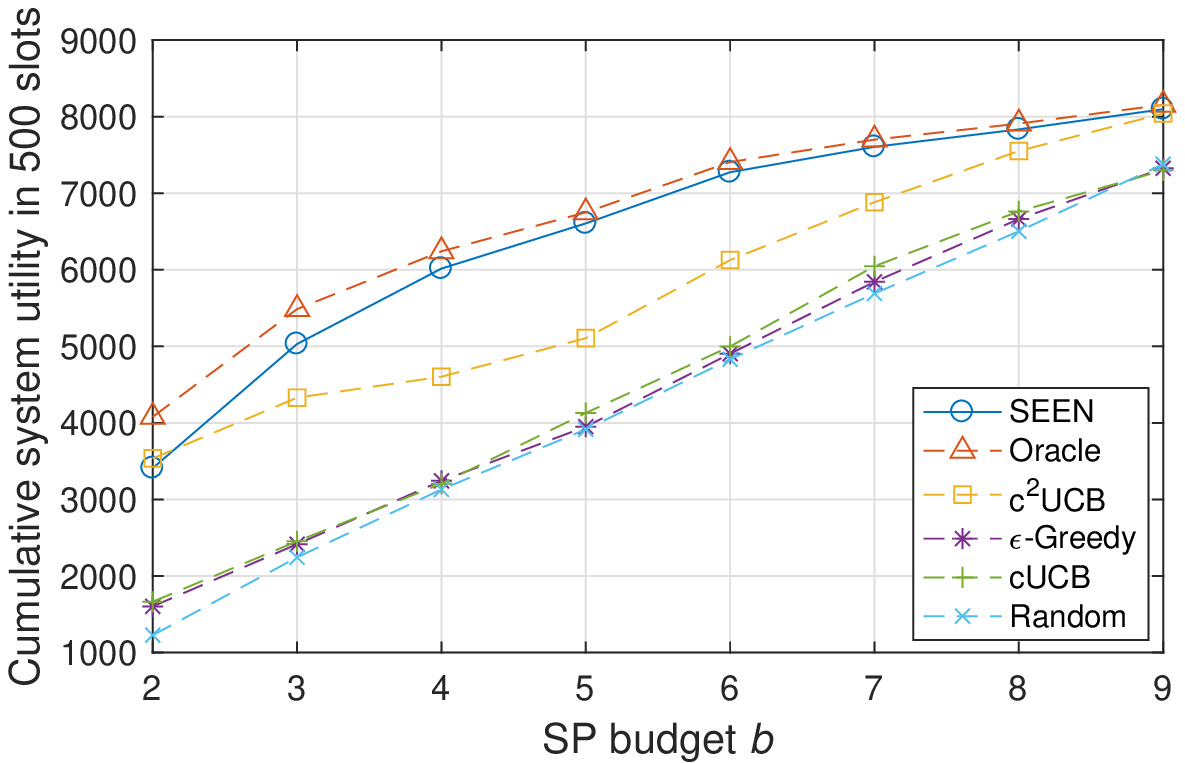}
		\vspace{-0.1 in}
		\caption{Impact of ASP budget.}
		\label{fig:varyb}
	\end{minipage}%
	\begin{minipage}[t]{0.5\linewidth}
		\centering
		\includegraphics[width=0.95\linewidth]{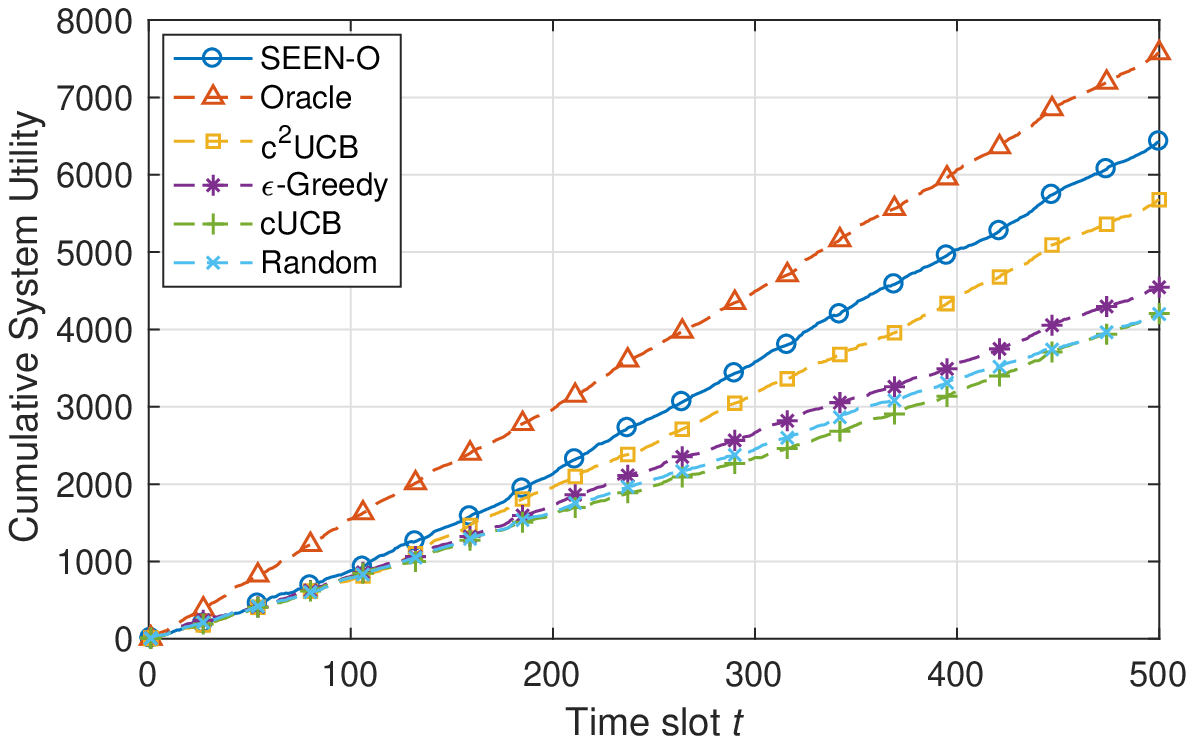}
		\vspace{-0.1 in}
		\caption{Comparison on cumulative utility (overlapped coverage).}
		\label{fig:cum_sys_utility_op}
	\end{minipage}%
	\vspace{-0.3 in}
\end{figure*}

\subsection{Edge Service Placement with Overlapping Coverage}
Fig. \ref{fig:cum_sys_utility_op} compares the performance achieved by SEEN-O and 5 other benchmarks when applied to the overlapping case. Similar to the non-overlapping case, we see that the context-aware schemes far outperform conventional MAB algorithms and SEEN-O achieves the highest cumulative system utility except for Oracle. However, it can be observed that SEEN-O incurs a larger regret compared to the non-overlapping case. This is because users in the overlapped area are observed by multiple SBSs and their contexts are duplicated when determining the under-explored SBSs. This increases the probability of being under-explored for SBSs and pushes SEEN-O to enter the exploration phase. Nevertheless, it does not mean that considering coverage overlapping leads to the performance degradation. SEEN-O actually achieves a higher cumulative system utility compared to that of SEEN achieved in the non-overlapping case.

\subsection{Impact of Overlapping Degree}
The overlapping degree of the edge network is defined as $S_{\text{co-cover}}/S_{\text{total}}$ where $S_{\text{co-cover}}$ is the service area co-covered by at least two SBSs and $S_{\text{total}}$ is the total service area. In the following, we show the impact of overlapping degree on the performance of SEEN. Fig. \ref{fig:overlap_deg} depicts the cumulative system utilities achieved by SEEN-O and Random in 500 time slots with different overlapping degrees. It also shows the cumulative system utility achieved by SEEN in the non-overlapping case for comparison. In general, we see that a larger overlapping degree results in higher system utilities for both SEEN-O and Random. This is because more users can access multiple SBSs for edge service given a larger overlapping degree, and therefore the ASP can further optimize the edge service placement decisions to accommodate more service demand at the Internet edge by exploiting the flexible association of users. By comparing SEEN-O and SEEN, we also see that considering the SBS coverage overlapping helps improve the system utility and the improvement grows with the increase in the overlapping degree.
\begin{figure}[htb]
	\centering
	\includegraphics[width=0.5\linewidth]{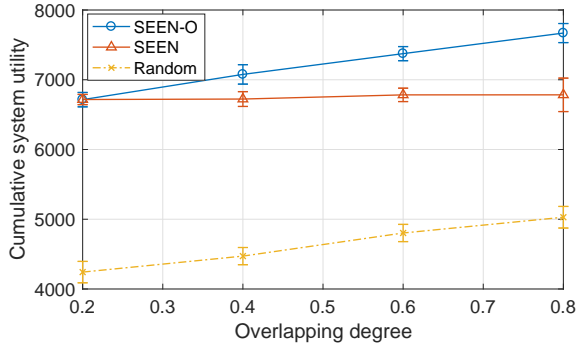}
	\caption{Impact of overlapping degree}
	\label{fig:overlap_deg}
	\vspace{-0.3 in}
\end{figure}

\section{Conclusion}\label{sec:conclusion}
In this paper, we investigated the edge service placement problem of an ASP in radio access networks integrated with shared edge computing platforms. To cope with the unknown and fluctuating service demand among changing user populations, we formulated a novel combinatorial contextual bandit learning problem and proposed an efficient learning algorithm to make optimal spatial-temporal dynamic edge service placement decisions. The proposed algorithm is practical, easy to implement and scalable to large networks while achieving provably asymptotically optimal performance. However, there are still efforts need to be done to improve the existing CC-MAB framework. First, we currently use a simple static partition of context space. Considering dynamic partition may further help improve the algorithm performance by generating more appropriate hypercubes. Second, our paper only provides a regret upper bound for SEEN. A meaningful complementary is to analyze the regret lower bound. Besides the investigated edge service placement problem, CC-MAB can also be applied to many other sequential decision making problems under uncertainty that involve multiple-play given a limited budget and context.

\bibliographystyle{IEEEtran}
\bibliography{refs}

\newpage
\appendices
\section{Proof of Theorem \ref{theo:regret_bound}}\label{proof:theorem_bound_R(T)}
The regret bound of SEEN is derived based on the natural assumption that the expected demands of users are similar if they have similar context as captured by the the H\"{o}lder condition. The H\"{o}lder condition allows us to derive a regret bound, which shows that the regret of SEEN is sublinear in the time horizon $T$, i.e. $R(T)=O(T^{\gamma})$ with $\gamma<1$.

For each SBS $n\in\mathcal{N}$ and each hypercube $p\in\mathcal{P}_{n,T}, \forall n$, we define $\bar{\mu}(p)=\sup_{x \in p}\mu(x)$ and $\ubar{\mu}(p)=\inf_{x \in p}\mu(x)$ be the best and worst expected demand over all contexts $x$ from hypercube $p$ respectively. In some steps of the proofs, we have to compare the demands at different positions in a hypercube. As a point of reference, we define the context at the (geometrical) center of a hypercube $p$ as $x^*(p)$. Also, we define the top-$b$ SBSs for hypercubes in $\P^t$ as following $b$ SBSs $\mathcal{S}^{*t}(\P^t)=\{n^*_1(\P^t),\dots, n^*_b(\P^t)\}$ which satisfy
\begin{align*}
&n^*_1(\P^t) \in \argmax_{n\in\mathcal{N}}\sum_{m \in \mathcal{M}^t_n}\tilde{u}_{n,m}\mu(x^*(p^t_{n,m}))\\
&n^*_2(\P^t) \in \argmax_{n\in\mathcal{N}\backslash\{n^*_1(\P^t)\}}\sum_{m \in \mathcal{M}^t_n}\tilde{u}_{n,m}\mu(x^*(p^t_{n,m}))\\
&\qquad\vdots\\
&n^*_b(\P^t) \in \argmax_{n\in\mathcal{N}\backslash\{n^*_1(\P^t),\dots,n^*_{b-1}(\P^t)\}}\sum_{m \in \mathcal{M}^t_n}\tilde{u}_{n,m}\mu(x^*(p^t_{n,m}))
\end{align*}
$\mathcal{S}^{*t}(\P^t)$ can be used to identify subsets of SBSs which are bad choices to rent when its users' contexts are from hypercubes $\p^t_n=(p^t_{n,m})_{m\in\mathcal{M}^t_n}$. Let
\begin{align}
&\mathcal{L}^t(\P^t)=\big\{ G=\{n_1,\dots ,n_b\} \subseteq \mathcal{N}, |G|=b: \nonumber \\
& \sum_{n\in S^{*,t}} \sum_{m \in \mathcal{M}^t_n} \tilde{u}_{n,m} \ubar{\mu}(p^t_{n,m}) - \sum_{n\in G} \sum_{m \in \mathcal{M}^t_n} \tilde{u}_{n,m} \bar{\mu}(p^t_{n,m}) \geq At^\theta \big\}
\end{align}
be the \emph{set of suboptimal subsets of SBSs} for the users' context $\P^t$, where $A>0$ and $\theta<0$ are parameters used only in the regret analysis. We call a subset $G$ of SBSs in $\mathcal{L}^t(\P^t)$ \emph{suboptimal} for $\P^t$, since the sum of the worst expected demands for $\mathcal{S}^{*t}(\P^t)$ is at least an amount $At^\theta$ higher than the sum of the best expected demands for subset $G$. We call subsets in $\mathcal{N}_b\backslash\mathcal{L}^t(\P^t)$ \emph{near-optimal} for $\P^t$. Here, $\mathcal{N}_b$ denotes the set of all $b$-element subsets of set $\mathcal{N}$. Then the regret $R(T)$ can be divided into the following three summands
\begin{align}
R(T)=\mathbb{E}[R_e(T)]+\mathbb{E}[R_s(T)]+\mathbb{E}[R_n(T)]
\end{align}
where the term $\mathbb{E}[R_e(T)]$ is the regret due to exploration phases and the term $\mathbb{E}[R_s(T)]$ and $\mathbb{E}[R_n(T)]$ are regrets in exploitation phases: the term $\mathbb{E}[R_s(T)]$ is the regret due to suboptimal choices, i.e., when subsets of SBSs from $\mathcal{L}^t(\P^t)$ are rented; the term $\mathbb{E}[R_n(T)]$ is the regret due to near-optimal choices, i.e., when subsets of SBSs from $\mathcal{N}_b\backslash\mathcal{L}^t(\P^t)$ are rented. Later, we will show that each of the three summands is bounded.

We first give the bound of $\mathbb{E}[R_e(T)]$ as shown in Lemma \ref{lemma:bound_R_e}.

\begin{lemma}[Bound for  $\mathbb{E}(R_e(T))$] \label{lemma:bound_R_e}
	Let $K_n(t)=t^{z_n}\log(t)$ and $h_{n,T}=\lceil T^{\gamma_n} \rceil$, where $0<z_n<1$ and $0<\gamma_n<\frac{1}{D_n}$. If SEEN is run with these parameters, the regret $\mathbb{E}[R_e(T)]$ is bounded by
	\begin{align}
	\mathbb{E}[R_e(T)]\leq b\tilde{u}^{\max}M^{\max}d^{\max}\sum_{n\in\mathcal{N}} 2^{D_n}\left(\log(T)T^{z_n+\gamma_nD_n} + T^{\gamma_nD_n}\right)
	\end{align}
\end{lemma}

\begin{proof}[Proof of Lemma \ref{lemma:bound_R_e}]
	Let $t$ be an exploration phase, then by the definition of SEEN, the set of under-explored SBSs $\mathcal{N}^{\text{ue},t}$ is non-empty in exploration, i.e., there exist a SBS $n$ and a hypercube $p^t_{n,m}, m\in\mathcal{M}^t_n$ with $C^t_{n}(p^t_{n,m})\leq K_n(t)=t^{z_n}\log(t)$. Clearly, there can be at most $\lceil T^{z_n} \log(T) \rceil$ exploration phases in which SBS $n$ is rented due to its under-exploration. Since there are $(h_{n,T})^{D_n}$ hypercubes in the partition, there can be at most $(h_{n,T})^{D_n}\lceil T^{z_n} \log(T) \rceil$ exploration phases in which SBS $n$ is rent due to its under-exploration. In each of these exploration phase, the maximum loss in demand due to wrong selection of a user in SBS $n$ is bounded by $\Delta^{\max}:= \max_{x\in\mathcal{X}_n,x^\prime\in \mathcal{X}_{n^\prime}}|\mu(x)-\mu(x^\prime)|$. Notice the random demand $\mu(x)$ for any $x\in \mathcal{X}_n, n\in\mathcal{N}$, is bounded in $[0, d^{\max}]$, it holds that $\Delta^{\max} \leq d^{\max}$. Let $\tilde{u}^{\max}$ be the maximum achievable delay improvement for any SBS $n$ by completing a unit workload (a task) for user $m\in\mathcal{M}^t_n$. Since the maximum number of users can be served by an SBS per time slot is $M^{\max}$, the maximum loss for wrong selection of SBS $n$ is bounded by $\tilde{u}^{\max}M^{\max}d^{\max}$. Additionally, we have to take into account the loss due to exploitations in the case that the size of under-explored SBSs is smaller than $b$. In each of the exploration phases in which SBS $n$ is selected, if the size of $\mathcal{N}^{\text{ue},t}$ is smaller than $b$, the maximum additional loss is $(b-1)\tilde{u}^{\max}M^{\max}d^{\max}$. Therefore, in each of exploration phase in which $n$ is selected, the overall maximum loss due to wrong selection of SBS $n$ is $b\tilde{u}^{\max}M^{\max}d^{\max}$. Summing over all $n\in\mathcal{N}$ yields:
	\begin{align}
	\mathbb{E}[R_e(T)] & \leq b\tilde{u}^{\max}M^{\max}d^{\max}\sum_{n\in\mathcal{N}} (h_{n,T})^{D_n}\lceil T^{z_n}\log(T)\rceil\\
	& = b\tilde{u}^{\max}M^{\max}d^{\max}\sum_{n\in\mathcal{N}} (T^{\gamma_n})^{D_n}\lceil T^{z_n}\log(T)\rceil
	\end{align}
	Using $\lceil T^{\gamma_n}\rceil^{D_n} \leq (2T^{\gamma_n})^{D_n} =2^{D_n}T^{\gamma_nD_n}$, it holds
	\begin{align}
	\mathbb{E}[R_e(T)]\leq b\tilde{u}^{\max}M^{\max}d^{\max}\sum_{n\in\mathcal{N}} 2^{D_n}\left(\log(T)T^{z_n+\gamma_nD_n} + T^{\gamma_nD_n}\right)
	\end{align}
\end{proof}

Next, we give a bound for $\mathbb{E}[R_s(T)]$. This bound also depends on the choice of two parameters $z_n$ and $\gamma_n$ for each SBS. Additionally, a condition on these parameters has to be satisfied.

\begin{lemma}[Bound for  $\mathbb{E}(R_s(T))$] \label{lemma:bound_R_s}
	Let $K_n(t)=t^{z_n}\log(t)$ and $h_{n,T}=\lceil T^{\gamma_n} \rceil$, where $0<z_n<1$ and $0<\gamma_n<\frac{1}{D_n}$. If SEEN is run with these parameters, Assumption \ref{holder} holds true and the additional condition $2H(t)+\tilde{u}^{\max}M^{\max}\left(\sum_{n \in G}L_nD_n^{\frac{\alpha_n}{2}}h^{-\alpha_n}_{n,T}+\sum_{n \in \mathcal{S}^{*t}(\P^t)}L_nD_n^{\frac{\alpha_n}{2}}h^{-\alpha_n}_{n,T}\right) \leq At^\theta$  is satisfied for all $1\leq t\leq T$ where $H(t):=b\tilde{u}^{\max}M^{\max}d^{\max}t^{-z^{\min}/2}, z^{\min} = \min_n z_n$. Then the regret $\mathbb{E}[R_s(T)]$ is bounded by
	\begin{align}
	\mathbb{E}[R_s(T)] \leq  b^2 \tilde{u}^{\max} (M^{\max})^2 d^{\max} {N \choose b} \frac{\pi^2}{3}
	\end{align}
\end{lemma}
\begin{proof}[Proof of Lemma \ref{lemma:bound_R_s}]
	For $1 \leq t \leq T$, let $W(t)=\{\mathcal{N}^{\text{ue},t} =\emptyset\}$ be the even that slot $t$ is an exploitation phase. By the definition of $\mathcal{N}^{\text{ue},t}$, in this exploitation, it holds that $C^t_n(p^t_{n,m})>K_n(t)=t^{z_n}\log(t)$ for all $n\in\mathcal{N}$ and all $m\in\mathcal{M}^t_n$. Let $V_G(t)$ be the event that subset $G$ is rented at time slot $t$. Then, it holds that
	\begin{align}
	R_s(T)= & \sum_{t=1}^{T}\sum_{G\in\mathcal{L}^t(\P^t)} I_{\{V_G(t), W(t)\}} \times \nonumber\\
	& \left(\sum_{n\in S^{*t}(\X^t)}\sum_{m \in \mathcal{M}^t_n} \tilde{u}_{n,m}d(x^t_{n,m})-\sum_{n\in G}\sum_{m \in \mathcal{M}^t_n} \tilde{u}_{n,m}d(x^t_{n,m})\right)
	\end{align}
	where in each time step, for the set $\mathcal{S}^{*t}(\X)$, the loss due to renting a suboptimal subset $G\in\mathcal{L}^t(\P^t)$ is considered. In each of the summands, the loss is given by comparing the demand for SBSs in $\mathcal{S}^{*t}(\X)$ with the demand for SBSs in the rented suboptimal set. Since the maximum loss per SBS is bounded by $\tilde{u}^{\max}M^{\max}d^{\max}$, we have
	\begin{align}
	R_s(T)\leq b\tilde{u}^{\max}M^{\max}d^{\max}\sum_{t=1}^{T}\sum_{G\in\mathcal{L}^t(\P^t)} I_{\{V_G(t), W(t)\}}
	\end{align}
	and taking the exception, the regret is hence bounded by
	\begin{align}
	\mathbb{E}[R_s(T)] &\leq b\tilde{u}^{\max}M^{\max}d^{\max}\sum_{t=1}^{T}\sum_{G\in\mathcal{L}^t(\P^t)} \mathbb{E}\left[I_{\{V_G(t), W(t)\}}\right] \nonumber\\
	& = b\tilde{u}^{\max}M^{\max}d^{\max}\sum_{t=1}^{T}\sum_{G\in\mathcal{L}^t(\P^t)} \text{Prob}\left\{V_G(t), W(t)\right\}
	\end{align}
	
	In the event of $V_G(t)$, by the construction of the algorithm, this means especially that the estimated utility achieved by SBSs in $G$ is at least as high as the sum of estimated utility of SBSs in $\mathcal{S}^{*t}(\P)$, i.e., $\sum_{n\in G}\sum_{m\in\mathcal{M}^t_n} \tilde{u}_{n,m}\hat{d}(p^t_{n,m}) \geq \sum_{n\in \mathcal{S}^{*t}(\P^t)}\sum_{m\in\mathcal{M}^t_n} \tilde{u}_{n,m}\hat{d}(p^t_{n,m})$. Thus, we have
	\begin{align} \label{prob_subopt}
	&\text{Prob}\left\{V_G(t),W(t)\right\} \nonumber \\
	&\leq \text{Prob}\left\{\sum_{n\in G}\sum_{m\in\mathcal{M}^t_n} \tilde{u}_{n,m}\hat{d}(p^t_{n,m}) \geq \sum_{n\in \mathcal{S}^{*t}(\P^t)}\sum_{m\in\mathcal{M}^t_n}\tilde{u}_{n,m}\hat{d}(p^t_{n,m})\right\}
	\end{align}
	
	The event in the right-hand side of \eqref{prob_subopt} implies at lease one of the three following events for any $H(t)>0$:
	\begin{align*}
	E_1= \left\{\sum_{n\in G}\sum_{m\in\mathcal{M}^t_n} \tilde{u}_{n,m}\hat{d}(p^t_{n,m}) \geq \sum_{n\in G}\sum_{m\in\mathcal{M}^t_n} \tilde{u}_{n,m}\bar{\mu}(p^t_{n,m})+H(t), W(t)\right\}
	\end{align*}
	\begin{align*}
	E_2= \left\{\sum_{n\in \mathcal{S}^{*t}(\P^t)}\sum_{m\in\mathcal{M}^t_n} \tilde{u}_{n,m}\hat{d}(p^t_{n,m}) \leq \sum_{n\in \mathcal{S}^{*t}(\P^t)}\sum_{m\in\mathcal{M}^t_n} \tilde{u}_{n,m}\ubar{\mu}(p^t_{n,m})-H(t), W(t)\right\}
	\end{align*}
	\begin{align*}
	E_3 = & \left\{\sum_{n\in G}\sum_{m\in\mathcal{M}^t_n} \tilde{u}_{n,m}\hat{d}(p^t_{n,m}) \geq \sum_{n\in \mathcal{S}^{*t}(\P^t)}\sum_{m\in\mathcal{M}^t_n} \tilde{u}_{n,m}\hat{d}(p^t_{n,m}),\right. \\
	& \quad \sum_{n\in G}\sum_{m\in\mathcal{M}^t_n} \tilde{u}_{n,m}\hat{d}(p^t_{n,m}) < \sum_{n\in G}\sum_{m\in\mathcal{M}^t_n} \tilde{u}_{n,m}\bar{\mu}(p^t_{n,m})+H(t),\\
	&  \left.\sum_{n\in \mathcal{S}^{*t}(\P^t)}\sum_{m\in\mathcal{M}^t_n} \tilde{u}_{n,m}\hat{d}(p^t_{n,m}) > \sum_{n\in \mathcal{S}^{*t}(\P^t)}\sum_{m\in\mathcal{M}^t_n} \tilde{u}_{n,m}\ubar{\mu}(p^t_{n,m})-H(t), W(t)\right\}.
	\end{align*}
	
	Hence, we have for the original event in \eqref{prob_subopt}
	\begin{align}\label{ori_E123}
	\left\{\sum_{n\in G}\sum_{m\in\mathcal{M}^t_n} \tilde{u}_{n,m}\hat{d}(p^t_{n,m}) \geq \sum_{n\in \mathcal{S}^{*t}(\P^t)}\sum_{m\in\mathcal{M}^t_n}\tilde{u}_{n,m}\hat{d}(p^t_{n,m})\right\}\subseteq E_1 \cup E_2 \cup E_3
	\end{align}
	
	The probability of the three event $E_1$, $E_2$, and $E_3$ will be bounded separately. We start with $E_1$, recall that the best expected demand from SBS $n$ in set $p\in\mathcal{P}_{n,T}$ is $\bar{\mu}(p)=\sup_{x \in p} \bar{\mu}(x)$. Therefore, the expected modified utilities of SBSs in $G$ is bounded by
	\begin{align}
	&\sum_{n\in G}\sum_{m \in \mathcal{M}^t_n}\tilde{u}_{n,m} \mathbb{E}\left[\hat{d}(p^t_{n,m})\right]\\
	= & \sum_{n\in G}\sum_{m \in \mathcal{M}^t_n}\tilde{u}_{n,m}\mathbb{E}\left[\frac{1}{|\mathcal{E}^t_n(p^t_{n,m})|} \sum_{(\tau,k): x^\tau_{n,k}\in p^\tau_{n,m}, n\in\mathcal{S}^t} d(x^\tau_{n,k})\right]\\
	= & \sum_{n\in G}\sum_{m \in \mathcal{M}^t_n}\tilde{u}_{n,m}\frac{1}{|\mathcal{E}^t_n(p^t_{n,m})|} \underbrace{\sum_{(\tau,k): x^\tau_{n,k}\in p^\tau_{n,m}, n\in\mathcal{S}^t}}_{|\mathcal{E}^t_n(p^t_{n,m})| \text{summands}} \underbrace{\mu(x^\tau_{n,k})}_{\leq \bar{\mu}(p^t_{n,m})}\\
	\leq & \sum_{n\in G}\sum_{m \in \mathcal{M}^t_n}\tilde{u}_{n,m}\bar{\mu}(p^t_{n,m})
	\end{align}
	
	This implies
	\begin{align}
	&\text{Prob}\{E_1\}\nonumber\\
	&=\text{Prob}\left\{\sum_{n\in G}\sum_{m\in\mathcal{M}^t_n} \tilde{u}_{n,m}\hat{d}(p^t_{n,m}) \geq \sum_{n\in G}\sum_{m\in\mathcal{M}^t_n} \tilde{u}_{n,m}\bar{\mu}(p^t_{n,m})+H(t), W(t)\right\}\nonumber\\
	&\leq \text{Prob}\left\{\sum_{n\in G}\sum_{m\in\mathcal{M}^t_n} \tilde{u}_{n,m}\hat{d}(p^t_{n,m}) \geq \sum_{n\in G}\sum_{m\in\mathcal{M}^t_n} \tilde{u}_{n,m}\mathbb{E}\left[\hat{d}(p^t_{n,m})\right]+H(t), W(t)\right\}\nonumber\\
	&\leq \sum_{n\in G}\sum_{m\in\mathcal{M}^t_n} \text{Prob}\left\{\tilde{u}_{n,m}\hat{d}(p^t_{n,m}) \geq  \tilde{u}_{n,m}\mathbb{E}\left[\hat{d}(p^t_{n,m})\right]+\dfrac{H(t)}{bM^{\max}}, W(t)\right\}\nonumber\\
	&\leq \sum_{n\in G}\sum_{m\in\mathcal{M}^t_n} \text{Prob}\left\{\hat{d}(p^t_{n,m}) \geq  \mathbb{E}\left[\hat{d}(p^t_{n,m})\right]+\dfrac{H(t)}{b\tilde{u}^{\max}M^{\max}}, W(t)\right\}
	\end{align}
	where the second last step follows the fact that if it would hold for all $n \in G$ and all $m\in \mathcal{M}^t_n$ that $\tilde{u}_{n,m}\hat{d}(p^t_{n,m}) <  \tilde{u}_{n,m}\mathbb{E}\left[\hat{d}(p^t_{n,m})\right]+\frac{H(t)}{bM^{\max}}$, then the line before could not hold true. Now, applying Chernoff-Hoeffding bound \cite{hoeffding1963probability} (note that for each SBS $n$, the estimated demand per user is bounded by $d^{\max}$) and then exploiting that event $W(t)$ implies that at least $t^{z_n}\log(t)$ samples were drawn from each SBS in $G$, we get
	\begin{align}\label{eq:prob_E_1}
	\text{Prob}\{E_1\} &\leq \sum_{n\in G}\sum_{m\in\mathcal{M}^t_n} \text{Prob}\left\{\hat{d}(p^t_{n,m}) -\mathbb{E}\left[\hat{d}(p^t_{n,m})\right] \geq  \dfrac{H(t)}{b\tilde{u}^{\max}M^{\max}}, W(t)\right\}\nonumber\\
	&\leq \sum_{n\in G}\sum_{m\in\mathcal{M}^t_n} \exp\left(\dfrac{-2|\mathcal{E}_n^t(p^t_{n,m})|H(t)^2}{b^2(M^{\max})^2(\tilde{u}^{\max})^2(d^{\max})^2}\right)\nonumber\\
	&\leq \sum_{n\in G}\sum_{m\in\mathcal{M}^t_n} \exp\left(\dfrac{-2H(t)^2t^{z_n}\log(t)}{b^2(M^{\max})^2(\tilde{u}^{\max})^2(d^{\max})^2}\right)
	\end{align}
	
	Analogously, it can be proven for event $E_2$, that
	\begin{align}\label{eq:prob_E_2}
	\text{Prob}\{E_2\} &=\text{Prob}\left\{\sum_{n\in \mathcal{S}^{*t}(\P^t)}\sum_{m\in\mathcal{M}^t_n} \tilde{u}_{n,m}\hat{d}(p^t_{n,m}) \right.\\
	&\qquad\qquad \left.\leq \sum_{n\in \mathcal{S}^{*t}(\P^t)}\sum_{m\in\mathcal{M}^t_n} \tilde{u}_{n,m}\ubar{\mu}(p^t_{n,m})-H(t), W(t)\right\}\\
	&\leq \sum_{n\in \mathcal{S}^{*t}(\P^t)}\sum_{m\in\mathcal{M}^t_n} \exp\left(\dfrac{-2H(t)^2t^{z_n}\log(t)}{b^2(M^{\max})^2(\tilde{u}^{\max})^2(d^{\max})^2}\right)
	\end{align}
	
	To bound the event $E_3$, we first make some additional definitions. First, we rewrite the estimate $\hat{d}(p), p\in \mathcal{P}_{n,T}$ as follows:
	\begin{align}
	\hat{d}(p)&=\dfrac{1}{|\mathcal{E}_n^t(p)|}\sum_{(\tau,k):x^\tau_{n,k}\in p, n\in \mathcal{S}^t} d(x^\tau_{n,k})\\
	&= \dfrac{1}{|\mathcal{E}_n^t(p)|}\sum_{(\tau,k):x^\tau_{n,k}\in p, n\in \mathcal{S}^t} \mu(x^\tau_{n,k})+\epsilon^\tau_{n,k}
	\end{align}
	
	where $\epsilon^\tau_{n,k}$ denotes the deviation from the expected demand of user with context $x^\tau_{n,k}$ in time slot $t$ covered by SBS $n$. Additionally, we define the best and worst context for a SBS $n\in\mathcal{N}$ in a set $p\in\mathcal{P}_{n,T}, \forall n$, i.e., $x^{\text{best}}(p):=\argmax_{x \in p} \mu(x)$ and $x^{\text{worst}}(p):=\argmin_{x \in p} \mu(x)$, respectively. Finally, we define the best and worst achievable demand for SBS $n$ in set $p$ as
	\begin{align}
	d^{\text{best}}(p)= \dfrac{1}{|\mathcal{E}_n^t(p)|}\sum_{(\tau,k):x^\tau_{n,k}\in p, n\in \mathcal{S}^t} \mu(x^{\text{best}}(p))+\epsilon^\tau_{n,k} \label{d_best}\\
	d^{\text{worst}}(p)= \dfrac{1}{|\mathcal{E}_n^t(p)|}\sum_{(\tau,k):x^\tau_{n,k}\in p, n\in \mathcal{S}^t} \mu(x^{\text{worst}}(p))+\epsilon^\tau_{n,k} \label{d_worst}
	\end{align}
	
	By H\"{o}lder condition from Assumption \ref{holder}, since $x^{\text{best}}(p) \in p$ and only contexts from hypercube $p$
	are used for calculating the estimated demand $\hat{d}(p)$, it can be shown that
	\begin{align}
	d^{\text{best}}(p)-\hat{d}(p)\leq L_nD_n^{\frac{\alpha_n}{2}}h^{-\alpha_n}_{n,T}
	\end{align}
	holds. Analogously, we have
	\begin{align}
	\hat{d}(p)-d^{\text{worst}}(p) \leq L_nD_n^{\frac{\alpha_n}{2}}h^{-\alpha_n}_{n,T}
	\end{align}
	
	Apply these results to the SBSs in $G$ and $\mathcal{S}^{*t}(\P^t)$ by summing over the SBSs, we have
	\begin{align}\label{eq:G_best_esti}
	\sum_{n \in G}\sum_{m\in\mathcal{M}^t_n} \left(\tilde{u}_{n,m}d^{\text{best}}(p^t_{n,m})-\tilde{u}_{n,m}\hat{d}(p^t_{n,m})\right) \leq \tilde{u}^{\max}M^{\max}\sum_{n \in G}L_nD_n^{\frac{\alpha_n}{2}}h^{-\alpha_n}_{n,T}
	\end{align}
	\begin{align}\label{eq:S*_best_esti}
	\sum_{n \in \mathcal{S}^{*t}(\P^t)}\sum_{m\in\mathcal{M}^t_n} \left(\tilde{u}_{n,m}\hat{d}(p^t_{n,m})-\tilde{u}_{n,m}d^{\text{worst}}(p^t_{n,m})\right) \leq \tilde{u}^{\max}M^{\max}\sum_{n \in \mathcal{S}^{*t}(\P^t)}L_nD_n^{\frac{\alpha_n}{2}}h^{-\alpha_n}_{n,T}
	\end{align}
	Now the three components of event $E_3$ are considered separately. By the definition of $d^{\text{best},t}_{n}(p)$ and $d^{\text{worst},t}_{n}(p)$ in \eqref{d_best} and \eqref{d_worst}. The first component of $E_3$, it holds that
	\begin{align}\label{E_3_1}
	&\left\{\sum_{n\in G}\sum_{m\in\mathcal{M}^t_n} \tilde{u}_{n,m}\hat{d}(p^t_{n,m}) \geq \sum_{n\in \mathcal{S}^{*t}(\P^t)}\sum_{m\in\mathcal{M}^t_n} \tilde{u}_{n,m}\hat{d}(p^t_{n,m})\right\}\\
	\subseteq & \left\{\sum_{n\in G}\sum_{m\in\mathcal{M}^t_n} \tilde{u}_{n,m}d^{\text{best}}(p^t_{n,m}) \geq \sum_{n\in \mathcal{S}^{*t}(\P^t)}\sum_{m\in\mathcal{M}^t_n} \tilde{u}_{n,m}d^{\text{worst}}(p^t_{n,m})\right\}
	\end{align}
	
	For the second component, using \eqref{eq:G_best_esti}, we have
	\begin{align}\label{E_3_2}
	&\left\{\sum_{n\in G}\sum_{m\in\mathcal{M}^t_n} \tilde{u}_{n,m}\hat{d}(p^t_{n,m}) < \sum_{n\in G}\sum_{m\in\mathcal{M}^t_n} \tilde{u}_{n,m}\bar{\mu}(p^t_{n,m})+H(t)\right\} \nonumber\\
	\subseteq & \left\{\sum_{n\in G}\sum_{m\in\mathcal{M}^t_n} \tilde{u}_{n,m}d^{\text{best}}(p^t_{n,m})- \tilde{u}^{\max}M^{\max}\sum_{n \in G}L_nD_n^{\frac{\alpha_n}{2}}h^{-\alpha_n}_{n,T} < \sum_{n\in G}\sum_{m\in\mathcal{M}^t_n} \tilde{u}_{n,m}\bar{\mu}(p^t_{n,m})+H(t)\right\} \nonumber\\
	= & \left\{\sum_{n\in G}\sum_{m\in\mathcal{M}^t_n} \tilde{u}_{n,m}d^{\text{best}}(p^t_{n,m}) < \sum_{n\in G}\sum_{m\in\mathcal{M}^t_n} \tilde{u}_{n,m}\bar{\mu}(p^t_{n,m}) + \tilde{u}^{\max}M^{\max}\sum_{n \in G}L_nD_n^{\frac{\alpha_n}{2}}h^{-\alpha_n}_{n,T} +H(t)\right\}
	\end{align}
	
	For the third component, using \eqref{eq:S*_best_esti}, we have
	\begin{align}\label{E_3_3}
	&\left\{\sum_{n\in \mathcal{S}^{*t}(\P^t)}\sum_{m\in\mathcal{M}^t_n} \tilde{u}_{n,m}\hat{d}(p^t_{n,m}) > \sum_{n\in \mathcal{S}^{*t}(\P^t)}\sum_{m\in\mathcal{M}^t_n} \tilde{u}_{n,m}\ubar{\mu}(p^t_{n,m})-H(t)\right\} \nonumber\\
	\subseteq & \left\{\sum_{n\in \mathcal{S}^{*t}(\P^t)}\sum_{m\in\mathcal{M}^t_n} \tilde{u}_{n,m} d^{\text{worst}}(p^t_{n,m})+\tilde{u}^{\max}M^{\max}\sum_{n \in \mathcal{S}^{*t}(\P^t)}L_nD_n^{\frac{\alpha_n}{2}}h^{-\alpha_n}_{n,T}\right. \nonumber \\& \qquad\qquad\qquad\qquad\qquad\qquad \left.>\sum_{n\in \mathcal{S}^{*t}(\P^t)}\sum_{m\in\mathcal{M}^t_n} \tilde{u}_{n,m}\ubar{\mu}(p^t_{n,m})-H(t)\right\} \nonumber\\
	= & \left\{\sum_{n\in \mathcal{S}^{*t}(\P^t)}\sum_{m\in\mathcal{M}^t_n} \tilde{u}_{n,m}d^{\text{worst}}(p^t_{n,m}) >\sum_{n\in \mathcal{S}^{*t}(\P^t)}\sum_{m\in\mathcal{M}^t_n} \tilde{u}_{n,m}\ubar{\mu}(p^t_{n,m}) \right. \nonumber \\& \qquad\qquad\qquad\qquad\qquad\qquad \left. -\tilde{u}^{\max}M^{\max}\sum_{n \in \mathcal{S}^{*t}(\P^t)}L_nD_n^{\frac{\alpha_n}{2}}h^{-\alpha_n}_{n,T}-H(t)\right\}
	\end{align}
	
	Therefore, using \eqref{E_3_1}, \eqref{E_3_2} and \eqref{E_3_3}, the probability of event $E_3$ is bounded by
	\begin{align}\label{E_3_prob_bound}
	~\text{Prob}&\{E_3\}\nonumber\\
	\leq~\text{Prob}&\left\{W(t), \sum_{n\in G}\sum_{m\in\mathcal{M}^t_n} \tilde{u}_{n,m}d^{\text{best}}(p^t_{n,m}) \geq \sum_{n\in \mathcal{S}^{*t}(\X^t)}\sum_{m\in\mathcal{M}^t_n} \tilde{u}_{n,m}d^{\text{worst}}(p^t_{n,m}), \right.\nonumber\\
	&~\sum_{n\in G}\sum_{m\in\mathcal{M}^t_n} \tilde{u}_{n,m}d^{\text{best}}(p^t_{n,m}) < \sum_{n\in G}\sum_{m\in\mathcal{M}^t_n} \tilde{u}_{n,m}\bar{\mu}(p^t_{n,m}) \nonumber\\ & \qquad\qquad\qquad\qquad\qquad\qquad + \tilde{u}^{\max}M^{\max}\sum_{n \in G}L_nD_n^{\frac{\alpha_n}{2}}h^{-\alpha_n}_{n,T}+H(t), \nonumber\\
	&~\sum_{n\in \mathcal{S}^{*t}(\P^t)}\sum_{m\in\mathcal{M}^t_n} \tilde{u}_{n,m}d^{\text{worst}}(p^t_{n,m}) >\sum_{n\in \mathcal{S}^{*t}(\P^t)}\sum_{m\in\mathcal{M}^t_n} \tilde{u}_{n,m}\ubar{\mu}(p^t_{n,m}) \nonumber \\& \qquad\qquad\qquad\qquad\qquad\qquad \left. -\tilde{u}^{\max}M^{\max}\sum_{n \in \mathcal{S}^{*t}(\P^t)}L_nD_n^{\frac{\alpha_n}{2}}h^{-\alpha_n}_{n,T}-H(t)\right\}.
	\end{align}
	
	We want to find a condition under which the probability for $E_3$ is zero. For this purpose, it is sufficient to show that the probability for the right-hand side in \eqref{E_3_prob_bound} is zero. Suppose that the following condition is satisfied:
	\begin{align}\label{eq:condition_E3}
	2H(t)+\tilde{u}^{\max}M^{\max}\left(\sum_{n \in G}L_nD_n^{\frac{\alpha_n}{2}}h^{-\alpha_n}_{n,T}+\sum_{n \in \mathcal{S}^{*t}(\P^t)}L_nD_n^{\frac{\alpha_n}{2}}h^{-\alpha_n}_{n,T}\right) \leq At^\theta
	\end{align}
	
	Since $G\in\mathcal{L}^t(\P^t)$, we have $\sum_{n\in S^{*t}(\P^t)} \sum_{m \in \mathcal{M}^t_n} \tilde{u}_{n,m} \ubar{\mu}(p^t_{n,m}) - \sum_{n\in G} \sum_{m \in \mathcal{M}^t_n} \tilde{u}_{n,m}\bar{\mu}_n(p^t_{n,m}) \geq A t^\theta$, which together with \eqref{eq:condition_E3} implies that
	\begin{align*}
	&\sum_{n\in S^{*t}(\P^t)} \sum_{m \in \mathcal{M}^t_n} \tilde{u}_{n,m} \ubar{\mu}(p^t_{n,m}) - \sum_{n\in G} \sum_{m \in \mathcal{M}^t_n} \tilde{u}_{n,m} \bar{\mu}(p^t_{n,m})\\&~~-\left(2H(t)+\tilde{u}^{\max}M^{\max}\left(\sum_{n \in G}L_nD_n^{\frac{\alpha_n}{2}}h^{-\alpha_n}_{n,T}+\sum_{n \in \mathcal{S}^{*t}(\P^t)}L_n D_n^{\frac{\alpha_n}{2}}h^{-\alpha_n}_{n,T}\right)\right)\geq 0
	\end{align*}
	Rewriting yields
	\begin{align}\label{condition_rewr}
	& \sum_{n\in S^{*t}(\P^t)} \sum_{m \in \mathcal{M}^t_n} \tilde{u}_{n,m} \ubar{\mu}(p^t_{n,m}) - \tilde{u}^{\max}M^{\max}\sum_{n \in \mathcal{S}^{*t}(\P^t)}L_n D_n^{\frac{\alpha_n}{2}}h^{-\alpha_n}_{n,T} - H(t)\nonumber\\
	\geq & \sum_{n\in G} \sum_{m \in \mathcal{M}^t_n} \tilde{u}_{n,m} \bar{\mu}(p^t_{n,m})+\tilde{u}^{\max}M^{\max}\sum_{n \in G}L_nD_n^{\frac{\alpha_n}{2}}h^{-\alpha_n}_{n,T}+H(t) \geq 0
	\end{align}
	If \eqref{condition_rewr} holds true, the three components of the right-hand side in \eqref{E_3_prob_bound} cannot be satisfied at the same time: Combining the second and third component of \eqref{E_3_prob_bound} with \eqref{condition_rewr} yields $\sum_{n\in G} \sum_{m\in\mathcal{M}^t_n} \allowbreak \tilde{u}_{n,m}d^{\text{best}}(p^t_{n,m}) < \sum_{n\in \mathcal{S}^{*t}(\P^t)}\sum_{m\in\mathcal{M}^t_n}\tilde{u}_{n,m}d^{\text{worst}}(p^t_{n,m})$, which contradicts the first term of \eqref{E_3_prob_bound}. Therefore, under condition \eqref{eq:condition_E3}, it follows that $\text{Prob}\{E_3\}=0$.
	
	So far, the analysis was performed with respected to an arbitrary $H(t)>0$. In the remainder of the proof, we choose $H(t)=b\tilde{u}^{\max}M^{\max}d^{\max}t^{-z^{\min/2}}$, where $z^{\min}=\min_{n\in\mathcal{N}} z_n$. Then, using \eqref{eq:prob_E_1} and \eqref{eq:prob_E_2}
	\begin{align}
	\text{Prob}\{E_1\} \leq&  \sum_{n\in G}\sum_{m\in\mathcal{M}^t_n} \exp\left(\dfrac{-2H(t)^2t^{z_n}\log(t)}{b^2(M^{\max})^2(\tilde{u}^{\max})^2(d^{\max})^2}\right) \nonumber \\
	= &  \sum_{n\in G}\sum_{m\in\mathcal{M}^t_n} \exp\left(\dfrac{-2\left(b\tilde{u}^{\max}M^{\max}d^{\max}t^{-z^{\min/2}}\right)^2t^{z_n}\log(t)}{b^2(M^{\max})^2(\tilde{u}^{\max})^2(d^{\max})^2}\right) \nonumber \\
	= & \sum_{n\in G}\sum_{m\in\mathcal{M}^t_n} \exp\left(-2t^{(z_n-z^{\min})}\log(t)\right) \nonumber \\
	\leq & \sum_{n\in G}\sum_{m\in\mathcal{M}^t_n} \exp\left(-2\log(t)\right) \nonumber \\
	\leq & bM^{\max}t^{-2}
	\end{align}
	and analogously
	\begin{align}
	\text{Prob}\{E_2\} \leq bM^{\max}t^{-2}
	\end{align}
	
	To sum up, under condition \eqref{eq:condition_E3}, using \eqref{ori_E123}, the probability in \eqref{prob_subopt} is bounded by
	\begin{align}
	& \text{Prob}\left\{V_G(t),W(t)\right\} \nonumber \\
	\leq & \text{Prob}\left\{E_1\cup E_2\cup E_3\right\} \nonumber\\
	\leq & \text{Prob}\left\{E_1\right\} + \text{Prob}\left\{E_2\right\} + \text{Prob}\left\{E_3\right\} \nonumber\\
	\leq & 2bM^{\max}t^{-2}
	\end{align}
	
	\begin{align}
	\mathbb{E}[R_s(T)] \leq & b\tilde{u}^{\max}M^{\max}d^{\max}\times\sum_{t=1}^{T}\sum_{G\in\mathcal{L}^t(P^t)} \text{Prob}\left\{V_G^t, W(t)\right\} \nonumber \\
	\leq & b\tilde{u}^{\max}M^{\max}d^{\max}{N \choose b}\sum_{t=1}^{T}2bM^{\max}t^{-2} \nonumber \\
	\leq & b^2\tilde{u}^{\max} (M^{\max})^2 d^{\max}{N \choose b} \cdot 2\sum_{t=1}^{\infty}t^{-2} \nonumber \\
	\leq & b^2 \tilde{u}^{\max} (M^{\max})^2 d^{\max} {N \choose b} \frac{\pi^2}{3}
	\end{align}
	where ${N \choose b}$ is number of subsets of size $b$ in $\mathcal{N}$ and the value of Dirichlet series is inserted in the last step.
\end{proof}

Now we give a bound for $\mathbb{E}\left[R_n(T)\right]$.
\begin{lemma}[Bound for  $\mathbb{E}(R_n(T))$] \label{lemma:bound_R_n}
	Let $K_n(t)=t^{z_n}\log(t)$ and $h_{n,T}=\lceil T^{\gamma_n} \rceil$, where $0<z_n<1$ and $0<\gamma_n<\frac{1}{D_n}$. If SEEN is run with these parameters, Assumption \ref{holder} holds true, the regret $\mathbb{E}[R_n(T)]$ is bounded by
	\begin{align}
	\mathbb{E}[R_n(T)]\leq 3b\tilde{u}^{\max}M^{\max}L_{\tilde{n}}D_{\tilde{n}}^{\frac{\alpha_{\tilde{n}}}{2}}T^{1-\gamma_{\tilde{n}}\alpha_{\tilde{n}}}+\dfrac{A}{1+\theta}T^{1+\theta}
	\end{align}
	where $\tilde{n}=\argmax_n L_{n}D_{n}^{\frac{\alpha_{n}}{2}}h^{-\alpha_{n}}_{n,T}$.
\end{lemma}
\begin{proof}[Proof of Lemma \ref{lemma:bound_R_n}]
	For $1 \leq t \leq T$, consider the event $W(t)$ as in the previous proof. Recall that the subset of SBSs rented by SEEN in time slot $t$ is denoted by $\mathcal{S}^t$. The loss due to near-optimal subsets can be written as
	\begin{align}
	& R_n(T)=\sum_{t=1}^{T}I_{\{W(t),\mathcal{S}^t \in \mathcal{N}_b\backslash\mathcal{L}^t(\P^t)\}}\times\left( \sum_{n\in\mathcal{S}^{*t}(\X^t)}\sum_{m \in \mathcal{M}^t_n} \tilde{u}_{n,m}d(x^t_{n,m})-\sum_{n\in\mathcal{S}^{t}}\sum_{m \in \mathcal{M}^t_n} \tilde{u}_{n,m}d(x^t_{n,m})\right)
	\end{align}
	where in each time slot in which the selected subset $\mathcal{S}^t$ is near-optimal, i.e.,  $\mathcal{S}^t\in\mathcal{N}_b \backslash\mathcal{L}^t(\P^t)$, the loss is considered for renting $\mathcal{S}^t$ instead of $\mathcal{S}^{*t}(\X^t)$. Let $Q(t) = W(t)\cap\{\mathcal{S}^t\in\mathcal{N}_b \backslash \mathcal{L}^t(\P^t)\}$ denote the event of renting a near-optimal set of SBSs. Then, it follows for the regret by taking the expectation:
	\begin{align*}
	\mathbb{E}\left[R_n(T)\right]=\sum_{t=1}^{T}\mathbb{E}\left[I_{\{Q(t)\}}\times \left( \sum_{n\in\mathcal{S}^{*t}(\X^t)}\sum_{m \in \mathcal{M}^t_n} \tilde{u}_{n,m}d(x^t_{n,m}) -\sum_{n\in\mathcal{S}^{t}}\sum_{m \in \mathcal{M}^t_n} \tilde{u}_{n,m}d(x^t_{n,m})\right)\right]
	\end{align*}
	By the definition of conditional expectation, this is equivalent to
	\begin{align*}
	&\mathbb{E}\left[R_n(T)\right]\\
	=&\sum_{t=1}^{T}\text{Prob}\{Q(t)\}\cdot\mathbb{E}\left[ \sum_{n\in\mathcal{S}^{*t}(\X^t)}\sum_{m \in \mathcal{M}^t_n} \tilde{u}_{n,m}d(x^t_{n,m}) -\sum_{n\in\mathcal{S}^{t}}\sum_{m \in \mathcal{M}^t_n} \tilde{u}_{n,m}d(x^t_{n,m}) \mid Q(t)\right]\\
	=&\sum_{t=1}^{T}\mathbb{E}\left[ \sum_{n\in\mathcal{S}^{*t}(\X^t)}\sum_{m \in \mathcal{M}^t_n} \tilde{u}_{n,m}d(x^t_{n,m}) -\sum_{n\in\mathcal{S}^{t}}\sum_{m \in \mathcal{M}^t_n} \tilde{u}_{n,m}d(x^t_{n,m}) \mid Q(t)\right]
	\end{align*}
	
	Now, let $t$ be the time slot, where $Q(t)$ holds true, i.e., the algorithm enters an exploitation phase and $J\in\mathcal{N}_b\backslash\mathcal{L}^t(\P^t)$. By the definition of $\mathcal{N}^{\text{ue},t}$, in this case it holds that $C^t_n(p^t_{n,m})>K_n(t)=t^{z_n}\log(t)$ for all $n\in\mathcal{N}$ and all $m\in\mathcal{M}^t_n$. In addition, since $J\in\mathcal{N}_b \backslash\mathcal{L}^t(P^t)$, it holds
	\begin{align}
	\sum_{n\in S^{*t}(\P^t)} \sum_{m \in \mathcal{M}^t_n} \tilde{u}_{n,m} \ubar{\mu}(p^t_{n,m}) - \sum_{n\in J} \sum_{m \in \mathcal{M}^t_n} \tilde{u}_{n,m} \bar{\mu}(p^t_{n,m}) < At^\theta
	\end{align}
	To bound the regret, we have to give an upper bound on
	\begin{align}
	&\sum_{t=1}^{T}\mathbb{E}\left[ \sum_{n\in\mathcal{S}^{*t}(\X^t)}\sum_{m \in \mathcal{M}^t_n} \tilde{u}_{n,m}d(x^t_{n,m}) -\sum_{n\in J}\sum_{m \in \mathcal{M}^t_n} \tilde{u}_{n,m}d(x^t_{n,m}) \mid Q(t)\right]\\
	= & \sum_{t=1}^{T}\left( \sum_{n\in\mathcal{S}^{*t}(\X^t)}\sum_{m \in \mathcal{M}^t_n} \tilde{u}_{n,m}\mu(x^t_{n,m}) -\sum_{n\in J}\sum_{m \in \mathcal{M}^t_n} \tilde{u}_{n,m}\mu(x^t_{n,m})\right)
	\end{align}
	
	Applying H\"{o}lder condition several times yields
	\begin{align}
	&\sum_{t=1}^{T}\left( \sum_{n\in\mathcal{S}^{*t}(\X^t)}\sum_{m \in \mathcal{M}^t_n} \tilde{u}_{n,m}\mu(x^t_{n,m}) -\sum_{n\in J}\sum_{m \in \mathcal{M}^t_n} \tilde{u}_{n,m}\mu(x^t_{n,m})\right)\\
	\leq & \sum_{t=1}^{T}\left( \sum_{n\in\mathcal{S}^{*t}(\X^t)}\sum_{m \in \mathcal{M}^t_n} \tilde{u}_{n,m}\mu(x^*(p^t_{n,m})) \right.\nonumber\\ & \qquad \left. +b\tilde{u}^{\max}M^{\max}L_{\tilde{n}}D_{\tilde{n}}^{\frac{\alpha_{\tilde{n}}}{2}}h^{-\alpha_{\tilde{n}}}_{\tilde{n},T} -\sum_{n\in J}\sum_{m \in \mathcal{M}^t_n} \tilde{u}_{n,m}\mu(x^t_{n,m})\right)\\
	\leq & \sum_{t=1}^{T}\left( \sum_{n\in\mathcal{S}^{*t}(\P^t)}\sum_{m \in \mathcal{M}^t_n} \tilde{u}_{n,m}\mu_n(x^*(p^t_{n,m})) \right.\nonumber\\ & \qquad \left. + b\tilde{u}^{\max}M^{\max}L_{\tilde{n}}D_{\tilde{n}}^{\frac{\alpha_{\tilde{n}}}{2}}h^{-\alpha_{\tilde{n}}}_{\tilde{n},T}-\sum_{n\in J}\sum_{m \in \mathcal{M}^t_n} \tilde{u}_{n,m}\mu(x^t_{n,m})\right)\\
	\leq & \sum_{t=1}^{T}\left( \sum_{n\in\mathcal{S}^{*t}(\P^t)}\sum_{m \in \mathcal{M}^t_n} \tilde{u}_{n,m}\inf_{x \in p^t_{n,m}}\mu(x) \right.\nonumber\\ & \qquad \left. + 2b\tilde{u}^{\max}M^{\max}L_{\tilde{n}}D_{\tilde{n}}^{\frac{\alpha_{\tilde{n}}}{2}}h^{-\alpha_{\tilde{n}}}_{\tilde{n},T}-\sum_{n\in J}\sum_{m \in \mathcal{M}^t_n} \tilde{u}_{n,m}\mu(x^t_{n,m})\right)\\
	\leq & \sum_{t=1}^{T}\left( \sum_{n\in\mathcal{S}^{*t}(\P^t)}\sum_{m \in \mathcal{M}^t_n} \tilde{u}_{n,m}\inf_{x \in p^t_{n,m}}\mu(x) \right.\nonumber\\ & \qquad \left. + 3b\tilde{u}^{\max}M^{\max}L_{\tilde{n}}D_{\tilde{n}}^{\frac{\alpha_{\tilde{n}}}{2}}h^{-\alpha_{\tilde{n}}}_{\tilde{n},T}-\sum_{n\in J}\sum_{m \in \mathcal{M}^t_n} \tilde{u}_{n,m}\sup_{x \in p^t_{n,m}}\mu(x)\right)\\
	\leq & \sum_{t=1}^{T}\left( \sum_{n\in\mathcal{S}^{*t}(\P^t)}\sum_{m \in \mathcal{M}^t_n} \tilde{u}_{n,m}\ubar{\mu}(x)-\sum_{n\in J}\sum_{m \in \mathcal{M}^t_n} \tilde{u}_{n,m}\bar{\mu}(x)\right) \nonumber \\ & \qquad  + 3b\tilde{u}^{\max}M^{\max}L_{\tilde{n}}D_{\tilde{n}}^{\frac{\alpha_{\tilde{n}}}{2}}h^{-\alpha_{\tilde{n}}}_{\tilde{n},T}\\
	\leq & 3b\tilde{u}^{\max}M^{\max}L_{\tilde{n}}D_{\tilde{n}}^{\frac{\alpha_{\tilde{n}}}{2}}h^{-\alpha_{\tilde{n}}}_{\tilde{n},T}+At^{\theta}
	\end{align}
	where $\tilde{n}=\argmax_n L_{n}D_{n}^{\frac{\alpha_{n}}{2}}h^{-\alpha_{n}}_{n,T}$; the second inequality comes from the definition of $\mathcal{S}^{*t}(\P^t)$, i.e, for $\mathcal{S}^{*t}(\P^t)$, the sum of weighted expected demands at the centers of the hypercubes is at least as high as that of any other subset of SBSs. Using $h^{-\alpha^n}_{n,T}=\lceil T^{\gamma_n} \rceil^{-\alpha_n} \leq T^{-\gamma_n\alpha_n}$, we further have
	\begin{align}
	&\mathbb{E}\left[\sum_{n\in\mathcal{S}^{*t}(\X^t)}\sum_{m \in \mathcal{M}^t_n} \tilde{u}_{n,m}d(x^t_{n,m}) -\sum_{n\in J}\sum_{m \in \mathcal{M}^t_n} \tilde{u}_{n,m}d(x^t_{n,m}) \mid Q(t)\right]\nonumber\\
	\leq & 3b\tilde{u}^{\max}M^{\max}L_{\tilde{n}}D_{\tilde{n}}^{\frac{\alpha_{\tilde{n}}}{2}}T^{-\gamma_{\tilde{n}}\alpha_{\tilde{n}}}+At^{\theta}
	\end{align}
	
	Therefore, the regret can be bounded by
	\begin{align}
	\mathbb{E}[R_n(T)]\leq \sum_{t=1}^{T}\left(3b\tilde{u}^{\max}M^{\max}L_{\tilde{n}}D_{\tilde{n}}^{\frac{\alpha_{\tilde{n}}}{2}}T^{-\gamma_{\tilde{n}}\alpha_{\tilde{n}}}+At^{\theta}\right)
	\end{align}
	This implies $\mathbb{E}[R_n(T)]\leq 3b\tilde{u}^{\max}M^{\max}L_{\tilde{n}}D_{\tilde{n}}^{\frac{\alpha_{\tilde{n}}}{2}}T^{1-\gamma_{\tilde{n}}\alpha_{\tilde{n}}}+\dfrac{A}{1+\theta}T^{1+\theta}$.
\end{proof}
The over all regret is now bounded by applying the above Lemmas.
\begin{proof}[Proof of Theorem 1]
	First, let $K_n(t)=t^{z_n}\log(t)$ and $h_{n,T}=\lceil T^{\gamma_n} \rceil$, where $0<z_n<1$ and $0<\gamma_n<\frac{1}{D_n}$; let $H(t):=b\tilde{u}^{\max}M^{\max}d^{\max}t^{-z^{\min/2}}$; let the condition $2H(t)+\tilde{u}^{\max}M^{\max} \allowbreak \left(\sum_{n \in G}L_nD_n^{\frac{\alpha_n}{2}}h^{-\alpha_n}_T+\sum_{n \in \mathcal{S}^{*t}(P^t)}L_nD_n^{\frac{\alpha_n}{2}}h^{-\alpha_n}_T\right) \leq At^\theta$ be satisfied for all $1<t<T$. Combining the results of Lemma 1, 2, 3, the regret $R(T)$ is bounded by
	\begin{align}
	R(T) \leq 	& b\tilde{u}^{\max}M^{\max}d^{\max}\sum_{n\in\mathcal{N}} 2^{D_n}\left(\log(T)T^{z_n+\gamma_nD_n} + T^{\gamma_nD_n}\right)  \\
	& + b^2 \tilde{u}^{\max} (M^{\max})^2 d^{\max} {N \choose b} \frac{\pi^2}{3}\nonumber  + 3b\tilde{u}^{\max}M^{\max}L_{\tilde{n}}D_{\tilde{n}}^{\frac{\alpha_{\tilde{n}}}{2}}T^{1-\gamma_{\tilde{n}}\alpha_{\tilde{n}}}+\dfrac{A}{1+\theta}T^{1+\theta}
	\end{align}
	
	The summands contribute to the regret with leading orders $O(\log(T)T^{z_n+\gamma_nD_n})$, $O(T^{1-\gamma_{\tilde{n}}\alpha_{\tilde{n}}})$ and $O(T^{1+\theta})$. In order to balance the leading orders, we select the parameters $z_n, \gamma_n, A, \theta$ as following values $z_n=\frac{2\alpha_n}{3\alpha_n+D_n}\in (0,1), \gamma_n=\frac{z_n}{2\alpha_n}\in(0,\frac{1}{D_n}), \theta=-\frac{z^{\min}}{2}$, and $A=2b\tilde{u}^{\max}M^{\max}d^{\max}+2b\tilde{u}^{\max}M^{\max}L_{\tilde{n}}D_{\tilde{n}}^{\alpha_{\tilde{n}}/2}$. Note that the condition \eqref{eq:condition_E3} is satisfied with these values. The the regret $R(T)$ reduces to
	\begin{align}
	& R(T) \leq b\tilde{u}^{\max}M^{\max}d^{\max}\sum_{n\in\mathcal{N}} 2^{D_n}\left(\log(T)T^{\frac{2\alpha_n+D_n}{3\alpha_n+D_n}} + T^{\frac{D_n}{3\alpha_n+D_n}}\right)  \\
	& + b^2 \tilde{u}^{\max} (M^{\max})^2 d^{\max} {N \choose b} \frac{\pi^2}{3}\nonumber  + 3b\tilde{u}^{\max}M^{\max}L_{\tilde{n}}D_{\tilde{n}}^{\frac{\alpha_{\tilde{n}}}{2}}T^{\frac{2\alpha_{\tilde{n}}+D_{\tilde{n}}}{3\alpha_{\tilde{n}}+D_{\tilde{n}}}}+\dfrac{A}{1+\theta}T^{\left.\frac{2\alpha_n+D_n}{3\alpha_n+D_n}\right|_{n=\argmin_{n}z_n}}
	\end{align}
	Let $\bar{n}=\argmax_n \frac{2\alpha_n+D_n}{3\alpha_n+D_n}$, then the leading order is $O(bN\tilde{u}^{\max}M^{\max}d^{\max}2^{D_{\bar{n}}}  T^{\frac{2\alpha_{\bar{n}}+D_{\bar{n}}}{3\alpha_{\bar{n}}+D_{\bar{n}}}} \allowbreak \log(T))$.
\end{proof}

\section{Proof of Theorem \ref{theo:regret_bound_seen_O}}\label{proof:theorem_bound_seen_o}
\begin{proof}
	The main steps to prove the regret bound for SEEN-O is similar to that in SEEN. The regret $R(T)$ is divided in to three terms: the regret for exploitation phases $\mathbb{E}[R_e(T)]$, the regret for suboptimal choice $\mathbb{E}[R_s(T)]$, and the regret for near-optimal choices $\mathbb{E}[R_s(T)]$.
	
	For the exploration phase in overlapped SBSs, SEEN-O randomly select SBSs from under-explored SBSs. However, the under-explored SBSs are determined based on their observed users without considering the service availability at nearby SBSs, therefore it is possible that a chosen under-explored SBS will receive no service requests from observed users and no service demand can observed for counter and demand estimation update. This problem is solved by the updating scheme of SEEN-O where the observed service demand from a user can be use to update the demand estimation at all SBSs covering that user. Therefore, the number of times that an SBS will be explored is still bounded in the worst case. Therefore, the bound of $\mathbb{E}[R_e(T)]$ in Lemma \ref{lemma:bound_R_e} holds for SEEN-O. 
	
	The proof for the bound of $\mathbb{E}[R_e(T)]$ and $\mathbb{E}[R_n(T)]$ is the same as in Lemma \ref{lemma:bound_R_s} and Lemma \ref{lemma:bound_R_n}. Therefore, the regret bound of SEEN-O is the same as that of SEEN.
\end{proof}
\end{document}